\documentclass[sigconf]{acmart}

\usepackage{natbib}
\newcommand{\kc}[1]{{\color{black}#1}}

\AtBeginDocument{%
  \providecommand\BibTeX{{%
    \normalfont B\kern-0.5em{\scshape i\kern-0.25em b}\kern-0.8em\TeX}}}

\setcopyright{acmcopyright}
\copyrightyear{2018}
\acmYear{2018}
\acmDOI{XXXXXXX.XXXXXXX}

\acmConference[Conference acronym 'XX]{Make sure to enter the correct
  conference title from your rights confirmation emai}{June 03--05,
  2018}{Woodstock, NY}
\acmPrice{15.00}
\acmISBN{978-1-4503-XXXX-X/18/06}

\begin{document}


% \title{Optimized Cost Per Click in Online Advertising: \\ A Game-Theoretic Analysis}

\title{Optimized Cost Per Click in Online Advertising: \\ A Theoretical Analysis}

% \affiliation{%

\author{Kaichen Zhang$^{1}$,
Zixuan Yuan$^{2}$,
Hui Xiong$^{1,3,*}$}
% \thanks{$^*$Corresponding author.}

\affiliation{$^1$Artificial Intelligence Thrust, Hong Kong University of Science and Technology (Guangzhou) \country{} }
\affiliation{$^2$Financial Technology Thrust, Hong Kong University of Science and Technology (Guangzhou) \country{} }
\affiliation{$^3$Department of Computer Science and Engineering, Hong Kong University of Science and Technology \country{} }
\affiliation{kzhangbi@connect.ust.hk,zixuanyuan@hkust-gz.edu.cn,xionghui@ust.hk \country{}}

% \author{Kaichen Zhang}
% \affiliation{%
%   \institution{AI Thrust}
%   }
% \email{kzhangbi@connect.ust.hk}

% \affiliation{%

% \affiliation{%

% \affiliation{%

% \affiliation{%

% \affiliation{%

% \affiliation{%


\begin{abstract}
In recent years, Optimized Cost Per Click (OCPC) and Optimized Cost Per Mille (OCPM) have emerged as the most widely adopted pricing models in the online advertising industry. 
However, the existing literature has yet to identify the specific conditions under which these models outperform traditional pricing models like Cost Per Click (CPC) and Cost Per Action (CPA).
% \yzx{However, little attention is focused on identifying the specific conditions under which these models outperform traditional pricing models like Cost Per Click (CPC) and Cost Per Action (CPA).}
% To address this gap, this paper proposes a game-theoretic model to compare OCPC with CPC and CPA, taking into account out-site scenarios and outside options as two key factors.
% To address this gap, this paper \kc{proposes an economic model to theoretically compare OCPC with CPC and CPA}, taking into account out-site scenarios and outside options as two key factors.
To fill the gap, this paper builds an economic model that compares OCPC with CPC and CPA theoretically, which incorporates out-site scenarios and outside options as two key factors.
Our analysis reveals that OCPC can effectively replace CPA by tackling the problem of advertisers strategically manipulating conversion reporting in out-site scenarios where conversions occur outside the advertising platform. 
Furthermore, OCPC exhibits the potential to surpass CPC in platform payoffs by providing higher advertiser payoffs and consequently attracting more advertisers. 
However, if advertisers have less competitive outside options and consistently stay in the focal platform, the platform may achieve higher payoffs using CPC. 
Our findings deliver valuable insights for online advertising platforms in selecting optimal pricing models, and provide recommendations for further enhancing their payoffs. 
To the best of our knowledge, this is the first study to analyze OCPC from an economic perspective. 
\kc{Moreover, our analysis can be applied to the OCPM model as well.}

% Moreover, the analysis in this study can be applied to the OCPM model as well.
% To our knowledge, this is the first study to analyze OCPC \kc{theoretically}. Moreover, the analysis in this study can apply to the OCPM model as well.
\end{abstract}


%
% Keywords. The author(s) should pick words that accurately describe
% the work being presented. Separate the keywords with commas.



\settopmatter{printfolios=true}
\maketitle

\section{Introduction}
% Online display advertising has grown rapidly over the past few decades. According to the Interactive Advertising Bureau report \cite{iab2022}, online advertising revenue in the U.S. was \$209.7 billion in 2022, more than 30 times the 2002 figure. Meanwhile, online advertising pricing models continue to iterate with the rapid evolution of the online advertising industry, as shown in Figure~\ref{fig:timeline}.
% Over the past few decades, online advertising has experienced exponential growth.
Recent decades have witnessed remarkable growth in the online advertising industry.
According to Interactive Advertising Bureau~\cite{iab2022}, 
% the revenue in the U.S. reached \$209.7 billion in 2022
U.S. revenues roared to \$209.7 billion in 2022, which marks an extraordinary 30-fold increase over the 2002 figures. 
Alongside this rapid expansion, the pricing models in online advertising have also consistently evolved over the years to adapt to market demands and technological advancements, as demonstrated in Figure~\ref{fig:timeline}.

% Concurrently, the pricing models for online advertising have been continuously evolving 
% to keep pace with the fast-paced advancements in the industry, as illustrated in Figure~\ref{fig:timeline}.

% The pricing models determine how advertisers bid and pay for their ads to be displayed on advertising platforms. In the Cost Per Mille (CPM) model, advertisers bid for ad displays, with the highest bidder making a payment to the platform and securing the right to display its ad to a user. 
% At its core,
In essence, the pricing model aims to determine how advertisers bid and pay for their ads to be displayed on advertising platforms. 
For instance, the Cost Per Mille (CPM) model allows advertisers to bid for ad displays. Then, the advertiser who offers the highest bid price will be secured with the right to display its ad. 
% However, much like traditional advertising media, CPM may display ads to unrelated users. This inefficiency is well captured by the famous marketing quote,
However, similar to traditional advertising media, CPM is not cost-efficient as it often distributes ads to unrelated users, which is well captured by the famous marketing quote,
"Half the money I spend on advertising is wasted; the trouble is I don't know which half." 
% However, much like traditional advertising media, CPM may result in ads being displayed to unrelated users. This inefficiency is aptly captured by the famous marketing quote, "Half the money I spend on advertising is wasted; the trouble is I don't know which half." 
To address this issue, the performance-based pricing models, Cost Per Click (CPC) and Cost Per Action (CPA), were developed in quick succession. In the CPC model, advertisers bid for ad clicks and only pay when a user clicks their ad. 
The CPA model takes a step further, allowing advertisers to directly bid for a desired action (e.g. making a purchase), and only pay when the user completes the desired action. 
Each time the user completes the desired action, it results in a conversion.
% The CPA model takes this a step further, allowing advertisers to directly bid for a desired action, such as making a purchase, and only pay when the user completes the designated action, i.e., a conversion occurs.

% The pricing models vary in terms of what advertisers bid for and what they pay for during real-time bidding auctions. In Cost Per Mille (CPM), advertisers bid for ad impressions and pay for ad impressions. Similar to traditional advertising media, CPM may deliver ad impressions to unrelated consumers, as the famous marketing quote says "Half the money I spend on advertising is wasted; the trouble is I don't know which half". To this end, the performance-based pricing models, Cost Per Click (CPC) and Cost Per Action (CPA), were invented in succession. In CPC, advertisers bid for ad clicks and pay for ad clicks. In CPA, advertisers bid for and pay for conversions, i.e., product purchases.
% In CPA, which was considered the "Holy Grail" by Google, advertisers bid for and pay for conversions, i.e., product purchases.

Most recently, Facebook introduced Optimized Cost Per Click (OCPC) and Optimized Cost Per Mille (OCPM). 
Both models inherit the good property of CPA that allows direct bidding for a conversion, while advertisers are still obliged to pay when a click occurs or an ad is displayed, like CPC and CPM.
% \yzx{As the improved versions of CPA, two recent models Optimized Cost Per Click (OCPC) and Optimized Cost Per Mille (OCPM)
% allows direct bidding for a conversion, while advertisers are still obliged to pay when a click occurs or an ad is displayed, like CPC and CPM.}
% These models also allow advertisers to bid directly for a conversion like CPA, but they are required to pay when an ad is displayed or a click occurs, like CPC and CPM. 
Hence, OCPC and OCPM have now become the most widely adopted pricing models in the industry, embraced by major platforms such as Google, Meta, TikTok, and Taobao.
% OCPC and OCPM have become the most widely adopted pricing models in the industry, offering higher payoffs for advertising platforms. Major platforms such as Google, Meta, TikTok, and Taobao have embraced these models. 
In contrast, the CPA model, once hailed as the "Holy Grail" of online advertising, now faces restrictions or is outright banned in many cases. It has been relegated to a limited set of scenarios, often accompanied by qualification requirements.

% Most recently, Facebook invented Optimized Cost Per Mille (OCPM) and Optimized Cost Per Click (OCPC), in which advertisers bid for conversions but pay for impressions in OCPM or clicks in OCPC. Now they have become the most widely used pricing models in the industry for higher advertising platform payoffs, adopted by Google, Meta, Tictok, Taobao, etc. In contrast, CPA, which was considered the online advertising "Holy Grail", is either banned or limited in narrow scenarios with qualification requirements.

\begin{figure*}[!t]
\centering
\includegraphics[width=\textwidth]{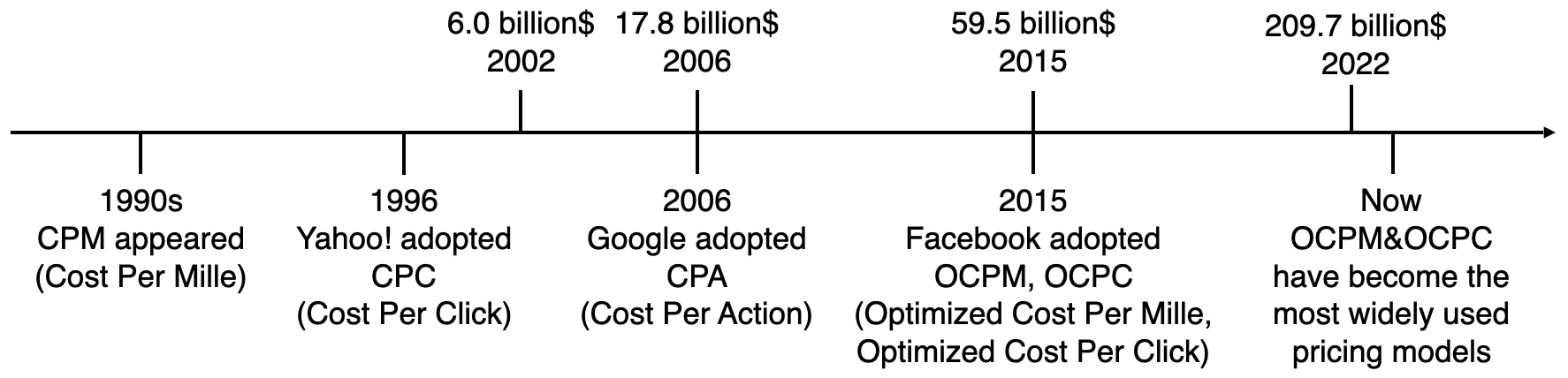}
% \vspace{-4mm}
\caption{U.S. Online Advertising Revenue and Pricing Model Milestones}\label{fig:timeline}
% \vspace{-5mm}
\end{figure*}

Despite their widespread adoption in practice, the existing literature has yet to identify the specific conditions under which OCPC and OCPM can outperform traditional models and the consequences of choosing them over alternatives. 
To fill this gap, we build an economic model to theoretically compare OCPC with traditional pricing models, CPC and CPA. Although our primary focus is on OCPC as an exemplar, it is worth noting that our analysis and insights can also apply to OCPM, considering their analogous structure. We include OCPM and CPM in Appendix \ref{app:other}.

% While we focus on OCPC as the representative pricing model, it is important to note that our analysis is also applicable to OCPM, given their similar position.
% Despite their popularity in practice, it is still unclear in the literature under which conditions OCPM and OCPC can outperform other pricing models and the consequences of adopting them over the other. To fill this gap, we propose a game-theoretic model to compare OCPC with traditional pricing models, CPC and CPA. While we focus on OCPC as the representative pricing model, it is important to note that our analysis is also applicable to OCPM, given their similar position.

% Despite their popularity in practice, it is still unclear in the literature under which conditions OCPM and OCPC can outperform other pricing models and the consequences of adopting them over the other. To this end, we propose a game-theoretic model to compare OCPC with traditional pricing models, CPC and CPA, given that OCPC shares the same payment part with CPC and the same bidding part with CPA. OCPC is selected as the representative pricing model, whose analysis also applies to OCPM. 

In our model, we solve an extensive-form game between one focal platform and two distinct types of advertisers. The game progresses through several stages. First, the platform selects a pricing model among CPC, CPA, or OCPC. Second, advertisers decide whether to enter the subsequent auction based on the selected pricing model. Third, a second-price auction ensues for an ad display. Finally, advertisers finalize payments and all parties realize payoffs.

% In our model, we solve a game between one focal online advertising platform and two advertisers with different types. First, the platform selects a pricing model among CPC, CPA, or OCPC. Second, advertisers decide whether to enter the following auction with the selected pricing model. Third, a second-price auction is conducted for the right to display an ad impression on the platform. Finally, advertisers settle payment and each party realizes payoffs.

Our model incorporates two crucial and realistic factors that previous pricing model research has overlooked. Firstly, we consider the fact that advertisers have alternative options beyond the focal platform. Advertisers, particularly those with budget constraints, may opt not to enter the auction to maximize return on investment (ROI) unless their payoffs exceed their outside options. On the other hand, some advertisers do not face budget constraints and can continue bidding for ads, such as paid video games, as long as their average cost remains lower than the marginal gain. Consequently, our model features one advertiser with an outside option and another without, representing these two types.

Furthermore, we emphasize the impact of conversion rate prediction algorithms. These platforms now employ machine learning to predict conversion rates for each user and ad, a feat difficult for advertisers to achieve independently. However, the accuracy of these predictions hinges on an unbiased training sample. This requirement may be compromised when conversions occur outside of platforms and depend on advertisers' reporting. As a result, we distinguish between in-site and out-site scenarios based on whether advertisers can manipulate conversion reporting.

% Our model incorporates two critical and realistic factors of online advertising that have been ignored in prior research on pricing models. First, advertisers face outside options other than the focal platform. Advertisers, especially those who have budget constraints, may not enter the auction unless their payoffs are greater than their outside options, because they have to decide how to spend their budget between different media and different advertising platforms to maximize their return on investment (ROI). On the other hand, some advertisers do not face budget constraints. As long as the average cost is less than the marginal gain, they can keep bidding for ads, like paid video games. Thus, one advertiser in our model has an outside option and another one does not, as the representative of the two types. Besides, we highlight the role of conversion rate prediction algorithms. Online advertising platforms now utilize machine learning to predict conversion rates for each user on each ad, which is hard to obtain by advertisers. However, the prediction is accurate only if the training sample is unbiased. This prerequisite is likely to be violated when conversions occur outside platforms and rely on advertisers' reporting. Thus, we differentiate the in-site scenario from the out-site scenario, depending on whether advertisers can manipulate conversion reporting.

Our model explains why OCPC outperforms CPA in practice. 
% This paper explains why OCPC outperforms CPA in practice. 
CPA is vulnerable to conversion reporting manipulation, resulting in negligible platform payoffs in out-site scenarios. 
In contrast, although advertisers can still manipulate conversion rates under OCPC, our analysis shows that they possess an incentive to maintain accurate conversion rates. 
There, OCPC functions effectively across a wider range of out-site scenarios.

Additionally, this paper highlights the trade-offs between OCPC and CPC. At first glance, platforms might always favor the so-called "optimized" pricing model, OCPC, over CPC. This preference stems from the potential for increased platform payoffs through OCPC, as it can attract more advertisers by offering them higher payoffs as well. However, if advertisers have less competitive outside options and consistently select the focal platform, the platform may actually achieve higher payoffs using CPC.

% This paper explains why OCPC dominates CPA in reality. CPA suffers from conversion reporting manipulation, leading to almost no platform payoffs in the out-site scenario. In contrast, though advertisers can still manipulate conversion rates, they have an incentive to report the true conversion rates in OCPC. As a result, OCPC works properly in the broader out-site scenarios. In addition, this paper sheds light on the trade-offs between OCPC and CPC. At first glance, platforms should always prefer the so-called "optimized" pricing model OCPC to CPC. Platforms can have higher payoffs using OCPC because OCPC can attract more advertisers by offering higher advertiser payoffs as well. However, if advertisers have less competitive outside options and consistently choose the focal platform, the platform may achieve higher payoffs using CPC. 

We further generalize our findings for the optimal pricing model selection among a set of pricing models. We suggest that platforms should choose a model that strikes a balance between platform payoffs with a predetermined number of advertisers and advertiser payoffs that attract more advertisers. When user value is high, platforms are advised to adopt models that yield higher platform payoffs. Conversely, platforms should opt for models that provide greater advertiser payoffs in an alternative manner.
% Our analysis suggests that platforms should balance between platform payoffs with a predetermined number of advertisers and advertiser payoffs that lure a greater number of advertisers. When user value is high, advertisers are more inclined to join the platform which in turn enables CPC or even CPM to boost platform payoffs further. Conversely, platforms should opt for OCPC to appeal to a larger pool of advertisers, consequently increasing their payoffs through an alternative approach.

% Our analysis suggests that platforms should balance platform payoffs with a fixed number of advertisers and advertiser payoffs that attract more advertisers. When user value is high, advertisers will naturally enter the platform, allowing CPC or even CPM to further increase platform payoffs. In contrast, platforms should choose OCPC to attract more advertisers and increase their payoffs in a different way.

% Our analysis indicates that platforms should balance their own payoffs and advertiser payoffs when selecting the optimal pricing model. When user value is high, platforms can select CPC or even CPM to increase their profits because advertisers face relatively lower outside options. On the other hand, platforms should choose OCPC to ensure advertisers enter their platform first.

This paper also offers suggestions for platforms to further enhance their payoffs. Firstly, platforms should consider shifting their attention from predicting conversion rates to actively increasing them, which can improve payoffs in a dual-pronged way. Secondly, platforms can employ subsidies or coupons to indirectly reduce outside options, subsequently attracting more advertisers. Lastly, platforms can develop innovative pricing models that allow advertisers to bid for aspects that fall between clicks and conversions to fill the gap between CPC and OCPC.

% Besides, this paper provides suggestions for platforms to further enhance their payoffs. First, platforms can reasonably switch their focus from conversion rate prediction to increasing conversion rate, which is shown to increase payoffs in a two-fold way. Second, platforms can use subsidies or coupons to decrease outside options indirectly, which attracts more advertisers in return. Third, platforms can develop new pricing models in which advertisers bid for something between clicks and conversions, which attracts more advertisers than CPC and achieves higher platform payoffs than OCPC.

In summary, this paper makes the following contributions:
\begin{itemize}
% \item To our knowledge, this is the first study to analyze OCPC from an economic perspective, thereby enriching the literature on online advertising pricing models.
\item To the best of our knowledge, this is the first study to analyze OCPC from an economic perspective, providing theoretical foundations to online advertising pricing models.
\item We have identified two critical factors that account for OCPC's prevalence over other pricing models: OCPC surpasses CPA as it is applicable to the broader \textit{out-site} scenarios; OCPC outperforms CPC by attracting a greater number of advertisers with \textit{outside options}.
\item Our findings provide valuable insights and recommendations for online advertising platforms to select the optimal pricing model and enhance their overall payoffs.
% \item To the best of our knowledge, this paper first provides theoretical explanations for OCPC and its comparisons with other pricing models, which adds to the literature on online advertising auctions.
% \item We identify two important factors that ensure OCPC's popularity over other pricing models. Specifically, OCPC can dominate CPA because OCPC applies to the broader out-site scenarios where advertisers can manipulate the reporting of conversion results; OCPC can dominate CPC because OCPC can attract more advertisers who have outside options by providing a higher advertiser payoff and social welfare.
% \item We identify two important factors that ensure OCPC's popularity over other pricing models. Specifically, OCPC can dominate CPA because OCPC applies to the broader \textit{out-site scenarios}; OCPC can outperform CPC because OCPC can attract more advertisers who have \textit{outside options}.

% \item We provide guidance on choosing the optimal pricing model for online advertising platforms and business implications for improving their payoff.

\end{itemize}

\section{Related Literature}
\kc{Information systems have been widely used in marketing \cite{petsinis2024robust, zhang2020geodemographic}, among which Real-time Bidding (RTB) in online advertising allows ads to be bought and sold on a per-impression (per-ad display) basis in real-time \cite{yuan2013real}. }
% Real-time bidding (RTB) is a contemporary sales channel in online advertising that allows ads to be bought and sold on a per-impression (per-ad display) basis in real-time \cite{yuan2013real}. 
This process is rooted in auction theories \cite{vickrey1961counterspeculation,clarke1971multipart,groves1973incentives,myerson1981optimal}, where an auction among advertisers is conducted as soon as a user lands on a website. Modern RTB predominantly employs the generalized second-price (GSP) auction model \cite{edelman2007internet,varian2007position}, in which the highest bidder secures the top ad placement and pays the second-highest bid. Similarly, the second-highest bidder acquires the second placement and pays the third-highest bid, and so forth. A comprehensive overview of online display advertising markets was provided by Choi et al. \cite{choi2020online}.
% Real-time bidding (RTB) \cite{yuan2013real} is a modern selling channel in online advertising where ads are bought and sold on a per-impression, i.e., per-display basis in real time. Based on auction theories \cite{vickrey1961counterspeculation,clarke1971multipart,groves1973incentives,myerson1981optimal}, an auction between advertisers is conducted immediately after a user arrives at the website for its impression. Contemporary RTB typically adopts the generalized second-price (GSP) auction\cite{edelman2007internet,varian2007position}, where the highest bidder gets the top ad placement and pays the bid of the second-highest bidder, the second-highest bidder gets the second ad placement and pays the bid of the third-highest bidder, and so on. Choi et al. \cite{choi2020online} provides a comprehensive summary of online display advertising markets.

% A line of research that is closely related to this paper employs \kc{economic} models to examine online advertising pricing models.
Economic models have been used to examine pricing models.
Hu \cite{hu2004performance} initially compared the CPM and CPC using the optimal contract design framework, while Asdemir et al. \cite{asdemir2012pricing} analyzed the same models under the principal-agent framework. Both studies argued that the performance-based CPC model can incentivize publishers and advertisers to enhance advertising performance. Agarwal et al. \cite{agarwal2009skewed} differentiated CPA from CPC. Hu et al. \cite{hu2016incentive} delved deeper into the trade-offs between CPA and CPC, focusing on incentive issues. Zhu et al. \cite{zhu2011hybrid} and Liu et al. \cite{liu2014information} investigated hybrid auctions that allow advertisers to select their preferred pricing models. 

% These papers typically employ stylized models featuring a single ad slot and two advertisers, enabling them to emphasize the main endogenous factors they aim to discuss.

% One stream of research utilizes analytical models to study the pricing model in online advertising, which closely relates to this paper. Hu \cite{hu2004performance} first compares the CPM and CPC pricing models under the optimal contract design framework. Asdemir et al. \cite{asdemir2012pricing} also compare the CPM and CPC pricing models but under the principal–agent framework. The above-mentioned two studies suggest CPC as a perform-based pricing model can provide incentives for both publishers and advertisers to improve advertising performance. Agarwal et al. \cite{agarwal2009skewed} distinguishes CPA, also as a perform-based pricing model, from CPC. Hu et al. \cite{hu2016incentive} further study the trade-offs between CPA and CPC, with a focus on the incentive issues. Besides, Zhu et al. \cite{zhu2011hybrid} and Liu et al. \cite{liu2014information} study hybrid auctions where advertisers can choose their pricing models. These papers usually adopt stylized models with one ad slot and two advertisers for convenience to emphasize the main role of the endogenous factors they wish to discuss.

However, existing literature on pricing models exhibits several gaps.
Firstly, to the best of our knowledge, no existing study has examined OCPC (or OCPM) as a pricing model. Although some papers \cite{zhu2017optimized,jain2021optimizing,tang2020optimized} have explored OCPC's algorithm system architecture, it is still unclear under which conditions OCPC should be employed in comparison to other pricing models.
Secondly, the extant literature on pricing models often assumes that advertisers always participate in the auction, thereby overlooking the significance of entry \cite{levin1994equilibrium,cao2013second,kaplan2022second,kirchkamp2009outside}.
% based on outside options.
This assumption may not accurately reflect real-world scenarios where advertisers decide whether to enter based on their outside options.
Thirdly, current research on pricing models does not consider the impact of conversion rate prediction algorithms. Lu et al. \cite{lu2017practical} highlighted that \kc{conversion trackers provided by platforms are just a piece of JavaScript code that can be arbitrarily excused or not by advertisers and}
online advertising platforms are no way to ascertain that they capture all the conversions without advertisers' truthful reports, which may considerably influence pricing models in a non-trivial manner.

\section{Preliminary}
In online advertising, an \textit{impression} occurs when an ad is displayed on a website or app. A \textit{click} occurs when a user clicks on the displayed ad. A \textit{conversion} occurs when a user completes a desired action, such as making a purchase, after clicking on the ad.

The \textit{CTR} (click-through rate) of an ad is the probability of its impression resulting in a click. The \textit{CVR} (conversion rate) of an ad is the probability of its click resulting in a conversion.

We define a pricing model based on what advertisers bid for and what they ultimately pay for. In \textit{CPC} (Cost Per Click), advertisers bid for clicks and pay for clicks. In \textit{CPA} (Cost Per Action), advertisers bid for conversions and pay for conversions. In \textit{OCPC} (Optimized Cost Per Click), advertisers bid for conversions but pay for clicks.

For instance, in the CPA model, an advertiser might bid \$10 for a conversion, signifying their willingness to pay up to \$10 for each conversion. If a conversion occurs, the advertiser has to pay the platform up to \$10. Alternatively, in the OCPC model, an advertiser can also bid \$10 for a conversion but it makes the payment when a click occurs. In this case, the payment could be up to \$1 per click, assuming the platform's predicted CVR is 0.1.

\kc{We categorize transactions that take place entirely within the ad platform as the \textit{in-site scenario}. Conversely, when users are directed from the platform to an external site after clicking on an ad, we refer to this as the \textit{out-site scenario}. To illustrate, Amazon belongs to the in-site scenario; Facebook belongs to the out-site scenario.}

% For example, in CPA, an advertiser may bid \$10 for a conversion, meaning they are willing to pay \$10 for a conversion. If a click becomes a conversion, the advertiser must pay up to \$10 to the platform. In another example, in OCPC, an advertiser can also bid \$10 for a conversion. But it must make its payment when a click occurs. The amount is up to 1\$ per click if the platform's prediction of the CVR is 0.1.

\section{Model}
We build an economic model of a game between two representative advertisers and one advertising platform using the second-price sealed bid auction under various pricing models.

% This analytical setting is a common choice in previous research on pricing models.
% Using a few representative advertisers is a common setting in previous research on pricing models. For instance, \cite{asdemir2012pricing} evaluates CPC and CPM by one agent advertiser; \cite{hu2016incentive} compares CPA and CPC by two advertisers with distinct types. The advantage of this setting is that it reveals the essence of the focused pricing models while significantly simplifying the analysis.\footnote{Google's Chief Economist Hal Varian advises: "An economic model is supposed to reveal the essence of what is going on: your model should be reduced to just those pieces that are required to make it work." \cite{varian2016build}} Finally, we discuss various extensions of our model in Appendix \ref{app:extensions}.}
The purpose of economic models is a simplification of and abstraction from complex reality. Previous research on pricing models typically uses a few representative advertisers to capture the essence of the focused pricing models while greatly simplifying the analysis.\footnote{Google's Chief Economist Hal Varian advises: "An economic model is supposed to reveal the essence of what is going on: your model should be reduced to just those pieces that are required to make it work." \cite{varian2016build}} For instance, \cite{asdemir2012pricing} evaluates CPC and CPM by one agent advertiser; \cite{hu2016incentive} compares CPA and CPC by two advertisers with distinct types. We follow this trend while discussing various extensions to more realistic settings in Appendix \ref{app:extensions}.

\begin{table}[t!]
    \centering
    \begin{tabular}{l l}
        \toprule
        \textbf{Notation} & \textbf{Meaning} \\
        \midrule
        $c^i,p^i$ & Advertiser $i$'s realized CTR,CVR \\
        $C^i,P^i$ & Advertiser $i$'s CTR,CVR distribution\\
        $\mu^c_i,\mu^p_i$ & The mean of $C^i,P^i$ \\[+1mm]
        $\bar{c^i},\bar{p^i}$ & The platform's prediction on $c^i,p^i$ \\
        $\bar{C^i},\bar{P^i}$ & Advertiser $i$'s belief in $\bar{c^i},\bar{p^i}$ \\
        $m^i$ & Advertiser $i$'s marginal gain from a conversion \\
        $b^i$ & Advertiser $i$'s bid \\
        $e^i$ & Advertiser $i$'s equivalent bid on an impression \\
        $\alpha^i$ & Advertiser $i$'s conversion reporting strategy \\
        \bottomrule
    \end{tabular}
    \vspace{+2mm}
    \caption{Notation Summary}
    \label{table:notation}
    \vspace{-10mm}
\end{table}

\subsection{Settings} \label{sec:setting}
\kc{We have clearly divided the settings into four groups: \textbf{Advertiser}, \textbf{CTR and CVR}, \textbf{In-site and Out-site scenario}, and \textbf{Auction}. Table~\ref{table:notation} summarizes this paper's main notations. In addition, we justify this paper's assumptions in Appendix \ref{app:assumption}.}

\noindent
\textbf{Advertiser:} 

We denote the marginal gain from a conversion for advertiser $i$ as $m^i$, indicating the profit obtained from each sale/conversion.
% We assume the marginal gain from a conversion for advertiser $i$ is represented as $m^i$, indicating the profit obtained from each sale.
% We assume the advertiser $i$'s marginal gain from a conversion as $m^i$, which indicates the profit it obtains from each sale. 

We assume advertiser 1 has no outside option, while advertiser 2 has an outside option $r$. Consequently, advertiser 1 always enters the auction and advertiser 2 enters the auction if and only if its expected payoffs in the auction surpass its outside option, $r$.

\kc{This asymmetry setting accounts for advertiser diversity. Advertiser 1 represents the group of advertisers without budget constraints, who will pursue conversions at a cost lower than marginal gain across different platforms. Advertiser 2 represents another group of advertisers with budget constraints, who will spend their budget on the most valuable platform to maximize their profit. In practice, the two groups of advertisers may prefer auto-bidders that maximize conversion volume or that manage budgets respectively.}

% In reality, one group of advertisers has no budget constraints and may prefer auto-bidders that maximize conversion volume \cite{heymann2019cost}. They are represented by advertiser 1. In contrast, another group of advertisers do have budget constraints and may use auto-bidders that manage budgets \cite{balseiro2017budget}, represented by advertiser 2.

% Advertiser 1 represents those advertisers who do not have a budget constraint for advertising. As long as their cost of obtaining one conversion is no greater than their marginal gain of one conversion, they can always participate in the auction for a non-negative profit. On the other hand, Advertiser 2 represents those advertisers with budget constraints. They only spend their budget on our focal platform if the profit is high enough. Even though they can have a positive profit in our focal platform, they will spend their budget on another platform if that platform has a higher profit.
\noindent
\textbf{CTR and CVR:} 

We represent advertiser $i$'s CTR as $c^i$ and CVR as $p^i$, drawn from the respective distributions $C^i$ and $P^i$. We define $E(C^i)=\mu^c_i, Var(C^i)=\sigma^c_i$ and $E(P^i)=\mu^p_i, Var(P^i)=\sigma^p_i$.

% We denote advertiser $i$'s CTR as $c^i$ and CVR as $p^i$ where $c^i$ and $p^i$ are drawn from the distribution $C^i$ and $P^i$ respectively. We denote $E(C^i)=\mu^c_i,Var(C^i)=\sigma^c_i$ and $E(P^i)=\mu^p_i,Var(P^i)=\sigma^p_i$.

Online advertising platforms now employ machine learning to predict CTR and CVR for each ad, a feat difficult for advertisers to achieve independently. Therefore, we assume that the platform's decisions are based on its predictions of $c^i$ and $p^i$, while advertisers make decisions based on the distribution of $C^i$ and $P^i$ and their belief in the accuracy of the platform's predictions.

% According to the online advertising frameworks, platforms make predictions on CTR and CVR for each ad. In contrast, advertisers do not and can not make such predictions. Thus we assume the platform makes decisions depending on its predictions on $c^i$ and $p^i$; the advertisers make decisions depending on the distribution of $C^i$ and $P^i$, and their belief in the platform's predictions.

We denote the platform's prediction for $c^i$ as $\bar{c^i}$ and $p^i$ as $\bar{p^i}$. We denote advertiser $i$'s belief in $\bar{c^i}$ as $\bar{C^i}$ and belief in $\bar{p^i}$ as $\bar{P^i}$. 

For example, consider advertiser 1's CTR follows a uniform distribution: $C^1\sim U(0.2,0.4)$. In a specific auction, the actual CTR could be 0.36, $c^1=0.36$. If the platform consistently underestimates the CTR by 50\%, then the predicted CTR would be $\bar{c^1}=0.5*0.36=0.18$. Although advertiser 1 is not aware of the actual CTR (0.36) or the platform's prediction (0.18), it can still form a belief about the platform's prediction, denoted as $\bar{C^1}\sim U(0.1,0.2)$.

% \textit{Example}. Suppose Advertiser 1's CTR follows a uniform distribution, $C^1\sim U(0.2,0.4)$. So the real CTR can be 0.36 in our focal auction, $c^1=0.36$. Suppose the platform always underestimates CTR by 50\%, then $\bar{c^1}=0.5*0.36=0.18$. Though Advertiser 1 does not know 0.36 or 0.18, it can form a belief of the platform's prediction, where $\bar{C^1}\sim U(0.1,0.2)$.

% We distinguish two scenarios, depending on whether the platform can exactly know the conversion result. In the in-site scenario, the platform can exactly know whether an impression converts or not. For instance, advertising on Amazon belongs to the in-site scenario because all transactions happen within Amazon. On the contrary, in the out-site scenario, the platform relies on advertisers to report their conversion results. For example, users on Facebook jump outside Facebook after clicking an ad. Facebook can hardly track conversions on advertisers' own websites.\footnote{Though online advertising platforms provide scripts for advertisers to embed in their websites, which automatically report conversions, the platforms can not ensure that advertisers always call the script when a conversion happens. Thus the platforms can not guarantee they can exactly know the conversion results by the embedded scripts.}
\noindent
\textbf{In-site and Out-site scenario:} 
% We differentiate between in-site and out-site scenarios based on whether advertisers can manipulate conversion reporting. For example, Amazon belongs to in-site scenarios as all transactions occur within its platform. Facebook belongs to out-site scenarios, as users are directed away from the platform after clicking an ad.

We distinguish between in-site and out-site scenarios based on whether advertisers can manipulate conversion reporting.

In the in-site scenario, we assume the platform's predictions are accurate, with $\bar{c^i}=c^i$ and $\bar{p^i}=p^i$, and advertisers' beliefs are consistent, with $\bar{C^i}=C^i$ and $\bar{P^i}=P^i$. In the out-site scenario, advertiser $i$ can establish a reporting strategy $\alpha^i \in [0,1]$, which signifies that a conversion is only reported with probability $\alpha^i$ upon receipt. We then assume the platform's predictions to be $\bar{c^i}=c^i$ and $\bar{p^i}=\alpha^ip^i$, with advertisers' beliefs at $\bar{C^i}=C^i$ and $\bar{P^i}=\alpha^iP^i$.

% We distinguish between in-site and out-site scenarios based on whether advertisers can manipulate conversion reporting. In the in-site scenario, we assume the platform's predictions are accurate: $\bar{c^i}=c^i, \bar{p^i}=p^i$. And the advertisers' beliefs are consistent: $\bar{C^i}=C^i, \bar{P^i}=P^i$. In the out-site scenario, advertiser $i$ should decide a reporting strategy $\alpha^i \in [0,1]$, indicating that it only reports a conversion with probability $\alpha^i$ when it receives a conversion. Correspondingly, we assume the platform's predictions are $\bar{c^i}=c^i, \bar{p^i}=\alpha^ip^i$, and the advertisers' beliefs are $\bar{C^i}=C^i, \bar{P^i}=\alpha^iP^i$. 

Our assumptions regarding the platform's predictions are based on the characteristics of real-world CTR/CVR supervised learning, making them suitable for analyzing pricing models. This approach takes into account the biased sampling error that may arise from advertisers' strategic conversion reports while excluding trivial learning errors to emphasize the primary role of pricing models.
% We assume how the platform makes predictions based on the nature of real-world CTR/CVR supervised machine learning. Our assumption is appropriate for the analysis of online advertising pricing models: it incorporates the biased sampling error potentially induced by advertisers' strategic conversion reports; it avoids trivial learning errors to highlight the main role of pricing models.

\noindent
\textbf{Auction:} 

We represent advertiser $i$'s bid as $b^i$. In CPC, $b^i$ denotes advertiser $i$'s bid per click, whereas, in CPA and OCPC, $b^i$ signifies advertiser $i$'s bid per conversion.

% We denote $b^i$ as the advertiser $i$'s bid. In CPC, $b^i$ indicates the advertiser $i$'s bid for a click; in CPA and OCPC, $b^i$ indicates the advertiser $i$'s bid for a conversion.

We introduce the concept of advertiser i's equivalent bid on an impression as $e^i$, which refers to the industry concept of ECPM \cite{zhu2017optimized}. The $e^i$ facilitates comparisons across various pricing models. For example, consider a scenario where $\bar{c^1}=0.3$, $\bar{p^1}=0.2$, and advertiser 1 bids 100 for a conversion, denoted as $b^1=100$. In this case, the equivalent bid is calculated as $e^1=0.3*0.2*100=6$. This suggests that, from the platform's viewpoint, advertiser 1 bidding 6 for an impression is equivalent. If advertiser 1 bids $10$ for a click, the equivalent bid becomes $e^1=10*0.3=3$. 

% The term $e^i$ can be related to the industry concept of ECPM.

% We define the Advertiser i's equivalent bid on an impression from the platform's perspective as $e^i$. This definition helps comparison between different pricing models. For instance, suppose $\bar{c^1}=0.3,\bar{p^1}=0.2$ and Advertiser 1 bids 100 for a conversion, $b^1=100$. In this case, $e^1=0.3*0.2*100=6$, which indicates that it is equivalent for Advertiser 1 to bid 6 for an impression from the platform's perspective. If Advertiser 1 bids 10 for a click, then $e^1=10*0.3=3$. The definition of $e^i$ may refer to the notion of eCPM in industry.

The platform selects the advertiser with the highest equivalent bid on an impression as the winner, denoted as $w=max_i e^i$. Another advertiser is deemed the loser, represented as $l=min_i e^i$.
% The platform will choose the advertiser with the maximum equivalent bid on an impression as the auction's winner. So the winner $w=max_i e^i$. And the advertiser who loses is $l=min_i e^i$ accordingly.

\subsection{Timeline}
The following timeline outlines the progression of the game:
\begin{enumerate}
    \item The platform selects a model among CPC, CPA, or OCPC.
    \item Advertiser 2 determines whether to enter the auction.
    \item Advertisers select their bids and, if in the out-site scenario, also decide on their reporting strategy $\alpha^i$.
    \item The actual values of $c^i$ and $p^i$ are realized, while the platform predicts $\bar{c^i}$ and $\bar{p^i}$, and calculates $e^i$.
    \item The platform identifies the auction winner ($w=argmax_i e^i$) and display its ad.
    \item The winner makes payments and all parties realize payoffs.
\end{enumerate}

We will analyze the game in both the in-site and out-site scenarios. The key difference between the two scenarios is that advertisers need to decide on a conversion reporting strategy in the out-site scenario (step 3), which also affects the payment process (step 6).

To solve the overall game, we will employ backward induction, focusing on each subgame where the platform has chosen a specific pricing model and both advertisers have opted to participate in the auction. We will begin our analysis from step 3 for each subgame.

% We will solve the game in both the in-site scenario and the out-site scenario. The two scenarios differ in that advertisers decide on a conversion reporting strategy in the out-site scenario (step 3) and thus the payment (step 6).

% We will use backward induction to solve the overall game by solving each subgame where the platform has chosen a specific pricing model and both advertisers have decided to participate in the auction, starting from step 3.

\section{in-site Scenario}
We initiate our analysis from the in-site scenario, in which advertising platforms can accurately track conversions. For each subgame, we will explore the optimal bidding strategy employed by advertisers and evaluate the resulting payoffs for all parties involved.

% We start our analysis from the in-site scenario, where online advertising platforms can exactly know conversion results. In each subgame, we will discuss the optimal bidding strategy of advertisers and payoffs for each party.

\subsection{Cost Per Click}
In CPC, advertisers bid for clicks and pay for clicks. 

As advertisers bid for clicks, the platform predicts advertiser $i$'s equivalent bid on an impression as $e^i=\bar{c^i}b^i$. Following the second-price auction, the platform selects the advertiser with the highest equivalent bid on an impression, denoted as $w=argmax_i e^i=argmax_i \bar{c^i}b^i$, to be the winner and receive the impression.

% As advertisers bid for clicks, the platform's prediction on advertiser $i$'s equivalent bid on an impression is $e^i=\bar{c^i}b^i$. According to the second-price auction, the platform chooses the advertiser with the highest equivalent bid on an impression as the winner, $w=argmax_i e^i=argmax_i \bar{c^i}b^i$, who receives the impression.

As advertisers pay for clicks, the auction winner has to pay $\frac{\bar{c^l}b^l}{\bar{c^w}}$ if a click occurs. The loser pays 0. This payment is based on the intuition that the equivalent payment for an impression is $\bar{c^l}b^l$, where $l$ denotes the advertiser with the second-highest bid in the model. However, since advertisers pay for a click in CPC, the auction winner should pay $\frac{\bar{c^l}b^l}{\bar{c^w}}$ for each click because every $\frac{1}{\bar{c^w}}$ impression generates a click on average.

% As advertisers pay for clicks, the auction winner has to pay $\frac{\bar{c^l}b^l}{\bar{c^w}}$ if a click occurs. The loser pays 0. The intuition behind the winner's payment is that the equivalent payment for an impression is $\bar{c^l}b^l$, whose $l$ represents the advertiser with the second price in our model. But because advertisers pay for a click in CPC, they should pay $\frac{\bar{c^l}b^l}{\bar{c^w}}$ for a click so that the payment for a click is equivalent to the payment for an impression because every $\frac{1}{\bar{c^w}}$ impressions generate a click on average.

Given another advertiser $k$'s equivalent bid on an impression $e^k$,
advertiser $i$'s utility is:
\begin{equation*}
% \vspace{-1mm}
\pi^i(b^i|e^k)
= Pr(\bar{C^i}b^i>e^k)E(C^iP^im^i-C^i \frac{e^k}{\bar{C^i}}|\bar{C^i}b^i>e^k)
% \vspace{-1mm}
\end{equation*}
when advertiser $i$ bid $b^i$ per click in CPC.

The first component, $Pr(\bar{C^i}b^i>e^k)$, represents advertiser $i$'s winning probability. The second component represents its expected gain upon winning. The $C^iP^im^i$ represents its revenue from the auctioned impression, while $C^i \frac{e^k}{\bar{C^i}}$ represents the associated cost.

\begin{proposition}
Under CPC, a dominant strategy for advertiser $i$ is to bid $b^i=\mu^p_im^i$, which is equal to the product of the expected conversion rate and the marginal gain from a conversion.
\end{proposition}

\begin{proof}
We show that $\pi^i(\mu^p_im^i|e^k)\ge\pi^i(b^i|e^k)$ for any $e^k$.
\begin{align*}
\pi^i(b^i|e^k)
= Pr(\bar{C^i}b^i>e^k)E(C^iP^im^i-C^i \frac{e^k}{\bar{C^i}}|\bar{C^i}b^i>e^k)
\\
=
Pr(\bar{C^i}b^i>e^k)E(C^i\mu^p_im^i-C^i \frac{e^k}{\bar{C^i}}|\bar{C^i}b^i>e^k)
\end{align*}
because $P^i$ is independent of the condition $\bar{C^i}b^i>e^k$.

When $b^i>\mu^p_im^i$, 
\begin{align*}
&\pi^i(b^i|e^k)
\\
&=Pr(\bar{C^i}\mu^p_im^i>e^k)E(C^i\mu^p_im^i-C^i \frac{e^k}{\bar{C^i}}|\bar{C^i}\mu^p_im^i>e^k)
\\
&+Pr(\bar{C^i}b^i >e^k>\bar{C^i}\mu^p_im^i)E(C^i\mu^p_im^i-C^i \frac{e^k}{\bar{C^i}}|\bar{C^i}b^i >e^k>\bar{C^i}\mu^p_im^i)
\\
&<=
Pr(\bar{C^i}\mu^p_im^i>e^k)E(C^i\mu^p_im^i-C^i \frac{e^k}{\bar{C^i}}|\bar{C^i}\mu^p_im^i>e^k) = \pi^i(\mu^p_im^i|e^k)
\end{align*}
because $E(C^i\mu^p_im^i-C^i \frac{e^k}{\bar{C^i}}|\bar{C^i}b^i >e^k>\bar{C^i}\mu^p_im^i)<0$

When $b^i<\mu^p_im^i$,
\begin{align*}
&\pi^i(\mu^p_im^i|e^k)
\\&=
Pr(\bar{C^i}b^i>e^k)E(C^i\mu^p_im^i-C^i \frac{e^k}{\bar{C^i}}|\bar{C^i}b^i>e^k)
\\&
+Pr(\bar{C^i}b^i <e^k<\bar{C^i}\mu^p_im^i)E(C^i\mu^p_im^i-C^i \frac{e^k}{\bar{C^i}}|\bar{C^i}b^i <e^k<\bar{C^i}\mu^p_im^i)
\\&
>=
Pr(\bar{C^i}b^i>e^k)E(C^i\mu^p_im^i-C^i \frac{e^k}{\bar{C^i}}|\bar{C^i}b^i>e^k)
= \pi^i(b^i|e^k)
\end{align*}
So $\pi^i(\mu^p_im^i|e^k)>=\pi^i(b^i|e^k)$ for any $e^k$.
\end{proof}

This optimal bidding strategy aligns with the truth-telling property in second-price auctions, where an advertiser's optimal bid is equal to their valuation of the auctioned impression. However, due to the lack of knowledge regarding the actual conversion rate ($p^i$), advertisers can only estimate the valuation using their average conversion rate ($\mu^p_i$). In practice, this average conversion rate can be calculated by dividing the total number of conversions by the total number of clicks.

% The optimal bidding is in line with the truth-telling property in second-price auctions. The advertiser's optimal bid equals the valuation of the auctioned impression. However, as advertisers do not know the realized conversion rate of the auctioned impression, $p^i$, they can just estimate the valuation of that impression by their average conversion rate, $\mu^p_i$. In practice, advertisers can obtain $\mu^p_i$ by dividing the number of people who completed the desired action by the number of people who clicked on the ad.

% We then derive the payoffs for each party in CPC. The $I-$ represents the in-site scenario.
Next, we determine the payoffs for each party in the CPC model, focusing on the in-site scenario, denoted by $I-$.

Advertiser $i$'s payoff, $\pi^i_{I-CPC}$: 
\begin{align*}
&E(C^i
P^im^i-C^k\mu^p_km^k|C^i
\mu^p_im^i>C^k
\mu^p_km^k)Pr(C^i
\mu^p_im^i>C^k
\mu^p_km^k)
\\ = &
E(C^i
\mu^p_im^i-C^k\mu^p_km^k|C^i
\mu^p_im^i>C^k
\mu^p_km^k)Pr(C^i
\mu^p_im^i>C^k
\mu^p_km^k)
  \end{align*}

Platform's payoff, $\pi^P_{I-CPC}$: 
\begin{align*}
&(E(C^2\mu^p_2m^2|C^1
\mu^p_1m^1>C^2
\mu^p_2m^2)Pr(C^1
\mu^p_1m^1>C^2
\mu^p_2m^2)
+\\
&E(C^1\mu^p_1m^1|C^1
\mu^p_1m^1<C^2
\mu^p_2m^2)Pr(C^1
\mu^p_1m^1<C^2
\mu^p_2m^2))
\\&=
E(min(C^1
\mu^p_1m^1,C^2
\mu^p_2m^2))
\end{align*}

Social welfare, $\pi^S_{I-CPC}=E(max(C^1
\mu^p_1m^1,C^2
\mu^p_2m^2))$.
% \[E(max(C^1
% \mu^p_1m^1,C^2
% \mu^p_2m^2))\]
\subsection{Cost Per Action}
In CPA, advertisers bid for conversions and pay for conversions.

As advertisers bid for conversions, the platform's prediction on advertiser $i$'s equivalent bid on an impression is $e^i=\bar{c^i}\bar{p^i}b^i$. 

As advertisers pay for conversions, the auction winner has to pay $\frac{\bar{c^l}\bar{p^l}b^l}{\bar{c^w}\bar{p^w}}$ if a conversion occurs. The loser pays 0. In the in-site scenario, the platform can accurately track conversions, which ensures a conversion will always result in a payment.

% The bid and payment in CPA have a similar format as CPC, by mathematically replacing $\bar{c}$ with $\bar{c}\bar{p}$, which bonds CTR and CVR together, referring to the concept of CTCVR (click-through\&conversion rate) in some literature.

Given another advertiser $k$'s equivalent bid on an impression $e^k$,
advertiser $i$'s utility is:
\begin{equation*}
\pi^i(b^i|e^k)= Pr(\bar{C^i}\bar{P^i}b^i>e^k)E(C^iP^im^i-C^iP^i \frac{e^k}{\bar{C^i}\bar{P^i}}|\bar{C^i}\bar{P^i}b^i>e^k)
\end{equation*}
when advertiser $i$ bid $b^i$ per conversion in CPA.

\begin{proposition} \label{pro:CPA}
Under CPA in the in-site scenario, a dominant strategy for advertiser $i$ is to bid $b^i=m^i$.
\end{proposition} 

The proof of Proposition \ref{pro:CPA} can be found in Appendix \ref{app:CPA}.

Proposition \ref{pro:CPA} reinforces the assertion that "CPA, or cost per action, is the Holy Grail for targeted advertising", because advertisers only need to bid their marginal gain from a conversion without any estimation on CTR and CVR. This is particularly important for fresh ad campaigns and advertisers who are new to the platform because there is no historical data for accurate estimation. These advertisers can circumvent the potential losses resulting from biased estimations, by directly bidding their marginal gain.

% Proposition \ref{pro:CPA} supports the quota, "CPA, or cost per action, is the Holy Grail for targeted advertising", because advertisers only need to bid their marginal gain from a conversion, $m^i$, without any estimation on CTR and CVR. This is particularly important for newcomer advertisers to the platform and new ads because they do not have any history of advertising records for estimation. A direct bid of the marginal gain avoids losses of biased estimation for them, which contributes to the platform's advertiser retention and thus profits.

We then examine the payoffs associated with CPA for each party.
% We then derive the payoffs for each party in CPA.

Advertiser $i$'s payoff, $\pi^i_{I-CPA}$: 
\begin{align*}
E(C^i
P^im^i-C^kP^km^k|C^i
P^im^i>C^kP^km^k)Pr(C^i
P^im^i>C^kP^km^k)&
\end{align*}

Platform's payoff, $\pi^P_{I-CPA}=E(min(C^1P^1m^1,C^2P^2m^2))$.
% \[E(min(C^1P^1m^1,C^2P^2m^2))\]

Social welfare, $\pi^S_{I-CPA}=E(max(C^1P^1m^1,C^2P^2m^2))$.
% \[E(max(C^1P^1m^1,C^2P^2m^2))\]
\subsection{Optimized Cost Per Click}
In OCPC, advertisers bid for conversions but pay for clicks. 

As advertisers bid for conversions, the platform's prediction on advertiser $i$'s equivalent bid on an impression is $e^i=\bar{c^i}\bar{p^i}b^i$. 

As advertisers pay for clicks, the auction winner has to pay $\frac{\bar{c^l}\bar{p^l}b^l}{\bar{c^w}}$ if a click occurs. The loser pays 0.

In both CPA and OCPC, advertisers bid for conversions, resulting in an equivalent bid on an impression, denoted as $e^i$. However, OCPC's payment structure differs from CPA, as advertisers pay for clicks instead of conversions. OCPC's payment would be calculated as $\bar{p^l}b^l$ per click, rather than $b^l$ per conversion.

% For example, if the winner $w$ has the same CTR as the loser $l$, the payment would be calculated as $\bar{p^l}b^l$ per click, rather than $b^l$ per conversion.

% The equivalent bid on an impression, $e^i$, is the same as it is in CPA because advertisers bid for conversions in both cases. However, the payment in OCPC has a unique format than CPC and CPA in that the dimension of the payment is different than the dimension of the bid. For instance, suppose the winner $w$ has the same CTR as the loser $l$, the payment is thus $\bar{p^l}b^l$ per click instead of $b^l$ per conversion.

Given another advertiser $k$'s equivalent bid on an impression $e^k$,
advertiser $i$'s utility is:
\[\pi^i(b^i|e^k)= Pr(\bar{C^i}\bar{P^i}b^i>e^k)E(C^iP^im^i-C^i \frac{e^k}{\bar{C^i}}|\bar{C^i}\bar{P^i}b^i>e^k)\]
when advertiser $i$ bid $b^i$ per conversion in OCPC.

\begin{proposition} \label{pro:OCPC}
Under OCPC in the in-site scenario, a dominant strategy for advertiser $i$ is to bid $b^i=m^i$.
\end{proposition}
The proof of Proposition \ref{pro:OCPC} can be found in Appendix \ref{app:OCPC}. We then derive the payoffs for each party in OCPC.

Advertiser $i$'s payoff, $\pi^i_{I-OCPC}$: 
\begin{align*}
E(C^i
P^im^i-C^kP^km^k|C^i
P^im^i>C^kP^km^k)Pr(C^i
P^im^i>C^kP^km^k)
\end{align*}

Platform's payoff, $\pi^P_{I-OCPC}=E(min(C^1P^1m^1,C^2P^2m^2))$.
% \[E(min(C^1P^1m^1,C^2P^2m^2))\]

Social welfare, $\pi^S_{I-OCPC}=E(max(C^1P^1m^1,C^2P^2m^2))$.
% \[E(max(C^1P^1m^1,C^2P^2m^2))\]

Focusing on the in-site scenario, we observe that advertisers employ a similar bidding strategy in OCPC as in CPA, wherein they conveniently bid their marginal gain. Moreover, the payoffs for each party remain consistent across both pricing models. We will provide a comprehensive analysis of these aspects in Section \ref{sec:analysis}, following our discussion of the out-site scenario in Section \ref{sec:outsite}.

% Here, in the in-site scenario, we just note that advertisers have the same bidding strategy in OCPC as CPA by conveniently bidding their marginal gain. Besides, the payoffs for each party are the same in both pricing models. We will conduct our analysis thoroughly in Section 7 after discussing the out-site scenario in Section 6.
\section{out-site Scenario} \label{sec:outsite}
In this section, we discuss CPC, CPA, and OCPC in the out-site scenario, where advertisers can manipulate conversion reporting.
% In this section, we analyze the out-site scenario, where online advertising platforms rely on advertisers to report their conversion results. The main difference between the in-site and out-site scenarios is that advertisers can affect the platform's CVR prediction and payment by strategically reporting only a part of all conversions. 

% We will also discuss CPC, CPA, and OCPC in the out-site scenario.

\subsection{Cost Per Click}
Recall that in CPC, advertisers bid for clicks and pay for clicks. The auction winner, denoted as $w = \text{argmax}_i \bar{c^i}b^i$, receives the auctioned impression and has to pay $\frac{\bar{c^l}b^l}{\bar{c^w}}$ if a click occurs.

% Recall that advertisers bid for clicks and pay for clicks in CPC. The auction winner is $w=argmax_i \bar{c^i}b^i$, who receives the auctioned impression and has to pay $\frac{\bar{c^l}b^l}{\bar{c^w}}$ if a click occurs.

It is important to note that CPC does not require online platforms to predict conversion rates. This simplifies the deployment of the CPC model, which likely explains their emergence before CPA.

% We note that CPC does not involve any CVR predictions for platforms. This means it is easier for online platforms to deploy a CPC mechanism because they do not need to predict CVR, which explains why CPC appears much earlier than CPA.

Additionally, in CPC, payments are made irrespective of conversions. The transaction between the platform and the winning advertiser concludes upon a click, whether or not it occurs. From that point onward, the advertiser bears the risk or benefit of a click eventually leading to a conversion.

% These characteristics result in the following proposition:

\kc{These characteristics show CPC is irrelevant to CVR predictions. As a result, CPC does not differ between the in-site and out-site scenarios that differ in CVR prediction manipulation: }

% Besides, the payment is regardless of conversions in CPC. The transaction between the platform and the advertiser who wins stops when a click occurs or not. After that, it's all the advertiser's own risk or benefit of whether a click finally becomes a conversion.

% Such properties lead to 

\begin{proposition}
Under CPC, there is no difference between the subgame in the in-site scenario and the subgame in the out-site scenario.
\end{proposition}

Advertiser $i$ also bids $b^i = \mu^p_im^i$ for a click. Furthermore, the payoffs for each party in the out-site scenario ($O-$) are identical to the payoffs in the in-site scenario ($I-$).

% Advertiser $i$ also bid $b^i=\mu^p_im^i$ for a click. Besides, the payoffs for each party in the out-site scenario ($O-$) are the same as the payoffs in the in-site scenario ($I-$):

Advertiser $i$'s payoff, $\pi^i_{O-CPC}$: 
\begin{align*}
&E(C^i
\mu^p_im^i-C^k\mu^p_km^k|C^i
\mu^p_im^i>C^k
\mu^p_km^k)Pr(C^i
\mu^p_im^i>C^k
\mu^p_km^k)
\end{align*}

Platform's payoff, $\pi^P_{O-CPC}=E(min(C^1
\mu^p_1m^1,C^2
\mu^p_2m^2))$.
% \begin{align*}
% E(min(C^1
% \mu^p_1m^1,C^2
% \mu^p_2m^2))
% \end{align*}

Social welfare, $\pi^S_{O-CPC}=E(max(C^1
\mu^p_1m^1,C^2
\mu^p_2m^2))$.
% \[E(max(C^1
% \mu^p_1m^1,C^2
% \mu^p_2m^2))\]
\subsection{Cost Per Action}
Recall that in CPA, advertisers bid for conversions and pay for conversions. The auction winner, represented by $w=argmax_i \bar{c^i}\bar{p^i}b^i$, receives the impression and must pay $\frac{\bar{c^l}\bar{p^l}b^l}{\bar{c^w}\bar{p^w}}$ if a conversion occurs.

% Recall that advertisers bid for conversions and pay for conversions in CPA. The auction winner is $w=argmax_i \bar{c^i}\bar{p^i}b^i$, who receives the auctioned impression and has to pay $\frac{\bar{c^l}\bar{p^l}b^l}{\bar{c^w}\bar{p^w}}$ if a conversion occurs.

% In the in-site scenario, the primary difference is that advertiser $i$ will select a reporting strategy $\alpha^i \in [0,1]$. This strategy indicates that the advertiser only reports a conversion with probability $\alpha^i$ when a conversion is received.

The main difference in the in-site scenario is that advertiser $i$ will select a reporting strategy $\alpha^i \in [0,1]$, indicating that it only reports a conversion with probability $\alpha^i$ when a conversion occurs.

The reporting strategy $\alpha^i$ impacts the game in two ways. First, by setting $\alpha^i<1$, the advertiser can make the payment less frequent. This creates an incentive for advertisers not to report conversions, as it can reduce their payment. However, the advertising platform will also adjust its CVR prediction to $\bar{p^i}=\alpha^ip^i$, which decreases the likelihood of advertiser $i$ winning the auction. As a result, the overall effect of conversion self-reporting remains unclear.

% The reporting strategy $\alpha^i$ affects the game in two ways. First, the advertiser can make the payment less frequent by setting $\alpha^i<1$. Intuitively, advertisers have an incentive not to report conversions in order to reduce their advertising fees. However, the advertising platform will also adjust its CVR prediction to $\bar{p^i}=\alpha^ip^i$, which reduces the winning probability of advertiser $i$ in the auction. Thus, it remains unclear what the overall effect of conversion self-reporting is.

Thus, given another advertiser $k$'s equivalent bid on an impression $e^k$,
advertiser $i$'s utility is:
\begin{equation*}
\pi^i(b^i,\alpha^i|e^k)= Pr(\bar{C^i}\bar{P^i}b^i>e^k)E(C^iP^im^i-C^i\alpha^iP^i \frac{e^k}{\bar{C^i}\bar{P^i}}|\bar{C^i}\bar{P^i}b^i>e^k)
\end{equation*}
when advertiser $i$ bid $b^i$ per conversion in CPA. Note that there is an additional $\alpha^i$ term in the out-site scenario.

\begin{proposition} \label{pro:i_CPA}
Under CPA in the out-site scenario, an equilibrium does not exist, as advertisers can always obtain a higher payoff by reducing their reporting strategy $\alpha^i$.
% Under CPA in the out-site scenario, there does not exist an equilibrium because advertisers can always achieve a higher utility by reducing their reporting strategy, $\alpha^i$.
\end{proposition}

\begin{proof}
We use proof by contradiction. If an equilibrium exists, in which the reporting strategy of advertiser $i$ is $\alpha^{i*}$. In this equilibrium, we have $\bar{p^i} = \alpha^{i*}p^i$ and $\bar{P^i} = \alpha^{i*}P^i$.

% Suppose there exists an equilibrium, where advertiser $i$'s reporting strategy is $\alpha^{i*}$. Then, there is $\bar{p^i}=\alpha^{i*}p^i$ and $\bar{P^i}=\alpha^{i*}P^i$ accordingly in this equilibrium.

If $\alpha^{i*} = 0$, then $\bar{p^i} = 0$ and $e^i = 0$. In this case, advertiser $i$ always loses the auction and attains a payoff of $0$. Consequently, advertiser $i$ can deviate to a positive $\alpha^{i*}$ to achieve a positive payoff.

% If $\alpha^{i*}=0$, then $\bar{p^i}=0$ and $e^i =0$. Advertiser $i$ will always lose the auction and have a $0$ utility. Thus, advertiser $i$ can deviate to a positive $\alpha^{i*}$ to have a positive utility.
If $\alpha^{i*} > 0$, we can deduce that for any $\alpha^i > 0$, the optimal $b_i = \frac{m^i}{\alpha^i}$, i.e., $\pi^i(\frac{m^i}{\alpha^i}, \alpha^i | e^k) \ge \pi^i(b^i, \alpha^i | e^k)$ for any $e^k$. The full proof is provided in Appendix \ref{app:i_CPA}. When $b_i = \frac{m^i}{\alpha^i}$, we have:

% If $\alpha^{i*}>0$,
% we derive that given an $\alpha^i>0$, the optimal $b_i=\frac{m^i}{\alpha^i}$, i.e., $\pi^i(\frac{m^i}{\alpha^i},\alpha^i|e^k)\ge\pi^i(b^i,\alpha^i|e^k)$ for any $e^k$. The full proof can be found in the appendix \ref{app:i_CPA}. Then, when $b_i=\frac{m^i}{\alpha^i}$,
\begin{align*}
&\pi^i(\alpha^i|e^k)
\\&= 
Pr(\bar{C^i}\bar{P^i}\frac{m^i}{\alpha^i}>e^k)E(C^iP^im^i-C^i\alpha^iP^i \frac{e^k}{\bar{C^i}\bar{P^i}}|\bar{C^i}\bar{P^i}\frac{m^i}{\alpha^i}>e^k)
\\&=
Pr(\bar{C^i}\alpha^{i*}P^i\frac{m^i}{\alpha^i}>e^k)E(C^iP^im^i-\frac{\alpha^i}{\alpha^{i*}}e^k|\bar{C^i}\alpha^{i*}P^i\frac{m^i}{\alpha^i}>e^k)
\end{align*}
Note that $\pi_i$ increases as $\alpha^i$ decreases. By deviating to $\alpha^i = \epsilon\alpha^{i*}$, where $0 < \epsilon < 1$, advertiser $i$ can achieve a higher payoff, which leads to a contradiction.
% Note that $\pi_i$ increases when $\alpha^i$ decreases.
% By deviating to $\alpha^i=\epsilon\alpha^{i*}$ where $0<\epsilon<1$, advertiser $i$ gets a higher utility.
\end{proof}

When advertisers adopt a fixed reporting strategy, regardless of the value of $\alpha^i$ (0.1 or even 0), the platform can compensate for this effect on payment by adjusting their CVR prediction to align with the reporting strategy, such that $\bar{p^i}=\alpha^ip^i$.

% Proposition \ref{pro:i_CPA}'s proof indicates that if advertisers choose a fixed reporting strategy, whatever $\alpha^i$ is (0.1 or even 0), the advertising platform can offset such effect for the payment by reducing their CVR prediction to accord with the reporting strategy, $\bar{p^i}=\alpha^ip^i$.

However, the consequence of CPA in the out-site scenario is highly detrimental for platforms. Advertisers will continually lower their reporting strategy to evade payment, causing the platform's revenue to approach zero. \kc{The platform cannot immediately detect or respond to the lowered reporting probability because the CVR prediction algorithm learning takes time.}

% However, the outcome of CPC in the out-site scenario turns out to be very negative. Advertisers can keep lowering their reporting strategy to avoid payment: advertisers' reporting probability approaches 0 and the platform's profit approaches 0. The intuition is that the machine learning algorithms for CVR prediction can not immediately converge to catch up with the newest lowered reporting strategy. 

Furthermore, the auction loses its effectiveness in selecting the highest value ad, as it devolves into a competition of who can submit the largest number, considering that $b_i=\frac{m^i}{\alpha^i}$. \kc{Consequently, both advertisers may have an equal 50\% chance of winning under the CPA model in the out-site scenario.} So we have:

% Besides, the auction loses its ability to select the ad with the highest value because the auction becomes a game of who can quote the bigger number of bids, given $b_i=\frac{m^i}{\alpha^i}$. Thus, we assume that both advertisers have a 50\% winning chance under CPA in the out-site scenario. So we have:

Advertiser $i$'s payoff, $\pi^i_{O-CPA}=\frac{1}{2}E(C^i
P^im^i)$.
% \[\frac{1}{2}E(C^i
% P^im^i)\]

Platform's payoff, $\pi^P_{O-CPA}=0$.
% \[\pi^P_{I-CPA} = 0\]

Social welfare, $\pi^S_{O-CPA}=\frac{1}{2}(E(C^1P^1m^1)+E(C^2P^2m^2))$.
% \[\frac{1}{2}(E(C^1P^1m^1)+E(C^2P^2m^2))\]

\begin{figure*}[!t]
\centering
\includegraphics[width=0.9\textwidth]{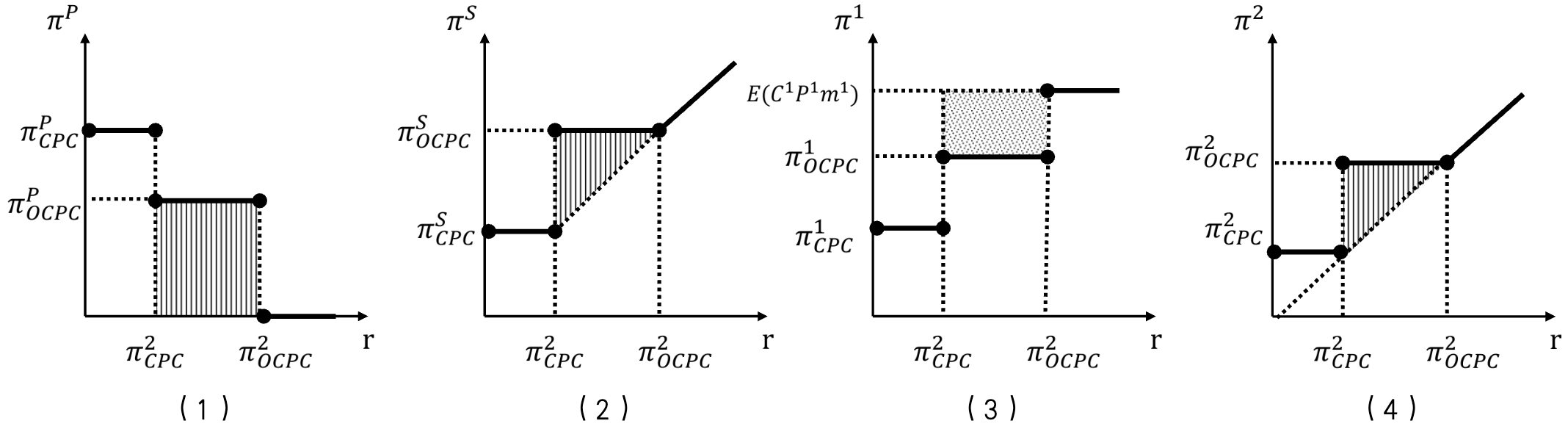}
% \vspace{-4mm}
\caption{Equilibrium Payoffs When the Outside Option $r$ is Varying}\label{fig:outside}
% \vspace{-5mm}
\end{figure*}

\subsection{Optimized Cost Per Click}
Recall that in OCPC, advertisers bid for conversions but pay for clicks. The auction winner, represented by $w=argmax_i \bar{c^i}\bar{p^i}b^i$, receives the impression and has to pay $\frac{\bar{c^l}\bar{p^l}b^l}{\bar{c^w}}$ if a click occurs.

Advertiser $i$ will also select a reporting strategy $\alpha^i$. We aim to investigate whether this reporting strategy also results in a dysfunctional OCPC, similar to the issues observed in CPA.

% In the out-site scenario under OCPC, advertiser $i$ will also choose a reporting strategy $\alpha^i \in [0,1]$, indicating that it only reports a conversion with probability $\alpha^i$ when it receives a conversion. We are wondering whether the conversion reporting also leads to a failed mechanism of OCPC as it is in CPA.

Given another advertiser $k$'s equivalent bid on an impression $e^k$,
advertiser $i$'s utility is:
\[\pi^i(b^i,\alpha^i|e^k)= Pr(\bar{C^i}\bar{P^i}b^i>e^k)E(C^iP^im^i-C^i \frac{e^k}{\bar{C^i}}|\bar{C^i}\bar{P^i}b^i>e^k)\]
when advertiser $i$ bid $b^i$ per conversion in OCPC.

\begin{proposition} \label{pro:i_OCPC}
Under OCPC in the out-site scenario, a set of dominant strategies for advertiser $i$ is to bid $\alpha^ib^i=m^i$.
\end{proposition}

\begin{proof}
We first show that $\alpha^i$ and $b^i$ always appear together. So we can represent $\pi^i(b^i,\alpha^i|e^k)$ by $\pi^i(\alpha^ib^i|e^k)$.
\begin{align*}
&\pi^i(b^i,\alpha^i|e^k)= Pr(\bar{C^i}\bar{P^i}b^i>e^k)E(C^iP^im^i-C^i \frac{e^k}{\bar{C^i}}|\bar{C^i}\bar{P^i}b^i>e^k)
\\&=
Pr(\bar{C^i}P^i\alpha^ib^i>e^k)E(C^iP^im^i-C^i \frac{e^k}{\bar{C^i}}|\bar{C^i}P^i\alpha^ib^i>e^k)
\\&=\pi^i(\alpha^ib^i|e^k)
\end{align*}
We can then deduce that $\pi^i(m^i|e^k)\ge\pi^i(\alpha^ib^i|e^k)$ for any $e^k$. The full proof can be found in Appendix \ref{app:i_OCPC}.
\end{proof}

Proposition \ref{pro:i_OCPC} suggests that the possible bidding and reporting strategy combinations include $\alpha^i=0.2,b^i=5m^i$ and $\alpha^i=0.5,b^i=2m^i$, among others. The truth-telling strategy, $\alpha^i=1,b^i=m^i$, is also part of their dominant strategies. Consequently, we have

% The Proposition \ref{pro:i_OCPC} states that advertisers will choose their bid and reporting strategy so that $\alpha^ib^i=m^i$. Those combinations include $\alpha^i=0.2,b^i=5m^i$ and $\alpha^i=0.5,b^i=2m^i$, etc. The truth-telling strategy, $\alpha^i=1,b^i=m^i$, also belongs to their dominant strategies. So we have

\begin{corollary} \label{cor:i_OCPC}
Truth-telling, i.e., bidding $b^i=m^i$ and $\alpha^i=1$, is a dominant strategy under OCPC in the out-site scenario.
\end{corollary}

Corollary \ref{cor:i_OCPC} implies that the game outcome of OCPC in the out-site scenario is identical to the in-site scenario, as the optimal bidding strategy for advertisers is always to bid their marginal gain, $m^i$. Although advertisers can manipulate conversion reporting in the out-site scenario, our proofs reveal that their optimal strategies consistently correspond to no manipulation. Hence, the payoffs for each party are as follows:

% The Corollary \ref{cor:i_OCPC} hints that the game outcome of OCPC in the out-site scenario is exactly the same as in the in-site scenario because the optimal bidding strategy for advertisers is always to bid their marginal gain, $m^i$. Though advertisers have the ability to manipulate conversion reporting in the out-site scenario, our proofs show that their optimal strategies are always equivalent to no manipulation. So the payoffs for each party are:

Advertiser $i$'s payoff, $\pi^i_{O-OCPC}$: 
\begin{align*}
E(C^i
P^im^i-C^kP^km^k|C^i
P^im^i>C^kP^km^k)Pr(C^i
P^im^i>C^kP^km^k)
\end{align*}

Platform's payoff, $\pi^P_{O-OCPC}=E(min(C^1P^1m^1,C^2P^2m^2))$.
% \[E(min(C^1P^1m^1,C^2P^2m^2))\]

Social welfare, $\pi^S_{O-OCPC}=E(max(C^1P^1m^1,C^2P^2m^2))$.
% \[E(max(C^1P^1m^1,C^2P^2m^2))\]

\section{Analysis} \label{sec:analysis}
% In this section, we first compare OCPC with CPA and OCPC with CPC. Then we derive the optimal pricing model for online advertising platforms. Finally, we discuss how advertisers' outside option affects the payoffs for each party and the optimal pricing model.

\subsection{Comparison between OCPC and CPA}
The main advantage of OCPC over CPA is its applicability to the out-site scenarios. In such scenarios, the platform's payoff under CPA, denoted as $\pi^P_{O-CPA}$, is zero due to conversion reporting manipulation. However, in OCPC, even though advertisers can manipulate conversion rates, our analysis indicates that it still yields positive platform payoffs, represented as $\pi^P_{O-OCPC} = E(min(C^1P^1m^1, C^2P^2m^2)) = \pi^P_{I-OCPC}>0$. This outcome is consistent with the in-site scenario as well. Furthermore, the social welfare in OCPC surpasses that in CPA, as evidenced by $\pi^S_{O-OCPC}>\pi^S_{O-CPA}$, since OCPC can effectively screen out the ad with the highest value in out-site scenarios, which CPA is unable to do.

% The main advantage of OCPC over CPA is that OCPC can apply to the out-site scenario while CPA cannot. In the out-site scenario, the platform's payoff in CPA is $\pi^P_{O-CPA} =0$ due to the manipulation of conversion reporting. In contrast, although advertisers can still manipulate conversion rates in OCPC, our analysis shows that OCPC still achieves positive platform payoffs, $\pi^P_{O-OCPC} = E(min(C^1P^1m^1, C^2P^2m^2)) = \pi^P_{I-OCPC}>0$, which is also same as in the in-site scenario. Moreover, the social welfare in OCPC is also greater than in CPA, $\pi^S_{O-OCPC}>\pi^S_{O-CPA}$, because OCPC can screen out the ad with the highest value in the out-site scenario, which CPA can't do.

Additionally, even in the in-site scenario, OCPC achieves identical results as CPA: both the platform's payoff, $\pi^p_{I-OCPC} =\pi^p_{I-CPA}$, the advertisers' payoffs, $\pi^i_{I-OCPC} =\pi^i_{I-CPA}$, and the social welfare, $\pi^S_{I-OCPC} =\pi^S_{I-CPA}$ are the same.

% In addition, even in the in-site scenario, OCPC achieves exactly the same results as CPA: the platform's payoff is the same. $\pi^p_{I-OCPC} =\pi^p_{I-CPA}$. The advertisers' payoffs are the same, $\pi^i_{I-OCPC} =\pi^i_{I-CPA}$. And the social welfare is the same, $\pi^S_{I-OCPC} =\pi^S_{I-CPA}$.

Our analysis explains why OCPC prevails over CPA in practice, given its broader applicability to out-site scenarios while maintaining the same performance in in-site scenarios as CPA. Consequently, OCPC serves as an ideal substitute for CPA.

% Our analysis explains why OCPC dominates CPA in reality because OCPC applies to the broader out-site scenario and achieves the same result in the in-site scenarios as CPA. As a result, OCPC is a perfect replacement for CPA. 

\subsection{Comparison between OCPC and CPC}
First, both OCPC and CPC exhibit the desirable property that the outcomes of the out-site scenario align with those of the in-site scenario. This property ensures their broad applicability across various situations.
% First, both OCPC and CPC have the good property that the out-site scenario has the same results as the in-site scenario. This property ensures that they are both widely applicable.

We proceed to compare the payoffs for each party between OCPC and CPC. We have the following lemma:

\begin{lemma} \label{lemma:social}
The social welfare in OCPC is greater than in CPC: 
\begin{align*} 
    \pi^S_{OCPC} =E(max(C^1P^1m^1, C^2P^2m^2))\\>
    \pi^S_{CPC} =E(max(C^1\mu^p_1m^1, C^2\mu^p_2m^2))
\end{align*}
\end{lemma}
The reason that OCPC yields higher social welfare than CPC lies in its increased efficiency with respect to ad delivery. In OCPC, the platform is responsible for predicting CVR and can employ advanced machine learning techniques to accurately determine the CVR and identify the most valuable ad. In contrast, CPC requires advertisers to estimate their CVR through crude averaging methods. However, we also have:

% OCPC has a higher social welfare than CPC because OCPC is more efficient in terms of ad delivery. In OCPC, it's the platform's job to predict CVR. They can use advanced machine learning techniques to accurately predict CVR and find the most valuable ad. But in CPC, advertisers can only roughly estimate their CVR by averaging. However, we also have:
\begin{lemma} \label{lemma:plat}
The platform's expected payoff in OCPC is less than in CPC: 
\vspace{-2mm}
\begin{align*} 
    \pi^P_{OCPC} =E(min(C^1P^1m^1, C^2P^2m^2))\\<
    \pi^P_{CPC} =E(min(C^1\mu^p_1m^1, C^2\mu^p_2m^2))
\end{align*}
\end{lemma}

The intuition behind Lemmas \ref{lemma:social} and \ref{lemma:plat} can be illustrated by considering the example of rolling two dice. The expected value of the maximum of these two dice (approximately $4.47$) should be greater than the maximum of the expected value of two dice ($\max{3.5, 3.5} = 3.5$). Conversely, the expected value of the minimum of these two dice (around $2.53$) should be less than the minimum of the expected value of two dice ($3.5$).

% The intuition behind the Lemma \ref{lemma:social} and \ref{lemma:plat} is like rolling two dice.  The expected value of the maximum of these two dice (about $4.47$) should be greater than the maximum of the expected value of two dice ($max\{3.5, 3.5\}=3.5$). Conversely, the expected value of the minimum of these two dice (about $2.53$) should be less than the minimum of the expected value of two dice ($3.5$).

At first glance, Lemma \ref{lemma:plat} may suggest that platforms should opt for CPC over OCPC, as CPC yields a higher platform payoff. However, this is not universally true. The current comparison between OCPC and CPC is based on the subgames in which advertisers have already decided to participate in the auction. Advertiser 2 possesses an outside option, $r$, and will choose not to enter the auction if the expected payoff is less than $r$. Furthermore, we have:

% Lemma \ref{lemma:plat} seems to imply that online advertising platforms should use CPC instead of OCPC because CPC has a higher payoff for the platform. However, this is not always the case. The current comparison between OCPC and CPC is based on the subgames where advertisers have already decided to enter the auction. As we assumed, advertiser 2 has an outside option, $r$, who will not enter the auction if its expected payoff is less than $r$. In addition, we have:

\begin{lemma} \label{lemma:adv}
The advertisers' expected payoff in OCPC is greater than in CPC: 
\begin{align*} 
    &\pi^i_{OCPC} =\\&E(C^i
P^im^i-C^kP^km^k|C^i
P^im^i>C^kP^km^k)Pr(C^i
P^im^i>C^kP^km^k)\\>
    &\pi^i_{CPC} =\\&E(C^i
\kc{P^im^i-C^kP^km^k}|C^i
% \mu^p_im^i-C^k\mu^p_km^k|C^i
\mu^p_im^i>C^k
\mu^p_km^k)Pr(C^i
\mu^p_im^i>C^k
\mu^p_km^k)
\end{align*}
\end{lemma}

Lemma \ref{lemma:adv} is based on the intuition that social welfare is the sum of platform payoff and advertiser payoff. CPC results in a lower social welfare and a higher platform payoff, which consequently leads to a lower advertiser payoff. \kc{The formal proofs of the three lemmas are given in Appendix \ref{app:lemma_social}, \ref{app:lemma_plat}, and \ref{app:lemma_adv}.}

% The intuition of Lemma \ref{lemma:adv} is that social welfare equals platform payoff plus advertiser payoff. CPC has a lower social welfare and a higher platform payoff, which leads to a lower advertiser payoff.

Lemma \ref{lemma:adv} suggests that if $\pi^2_{OCPC}>r>\pi^2_{CPC}$, advertiser 2 will enter the OCPC auction but not the CPC auction because its payoff in CPC is lower than its outside option. Consequently, advertiser 2's absence from the auction would lead to a decrease in platform payoffs under the CPC model.

% Lemma \ref{lemma:adv} suggests that if $\pi^2_{OCPC}>r>\pi^2_{CPC}$, advertiser 2 will enter the OCPC auction but not the CPC auction because its payoff in CPC is lower than its outside option. Its absence from the auction will therefore reduce the platform's overall payoffs if using CPC.

Our analysis illuminates the trade-offs between OCPC and CPC models. Platforms can achieve higher payoffs by adopting OCPC, as it can draw a larger number of advertisers by providing them with increased payoffs. However, if advertisers have less competitive alternative options and consistently opt for the focal platform, the platform may be able to attain higher payoffs using the CPC model.

% Our analysis sheds light on the trade-offs between OCPC and CPC. Platforms can have higher payoffs using OCPC because OCPC can attract more advertisers by offering higher advertiser payoffs as well. However, if advertisers have less competitive outside options and consistently choose the focal platform, the platform may achieve higher payoffs using CPC. 

\subsection{Optimal Pricing Model for Platforms}
We then generalize our findings for the optimal pricing model selection among a set of pricing models, by solving the second and first stages of our game.

In the second stage, given that the platform has selected pricing model $X$, advertiser 2 will participate in the auction if and only if its payoffs from the auction exceed its outside option $r$: 
% In the second stage, assuming the platform has chosen the pricing model $X$, advertiser 2 will enter the auction if and only if its expected payoffs in the auction are greater than the outside option $r$:
\[\pi^2_{X}>r\]
% \[\pi^2_{C-X}>r\ or\ \pi^2_{I-X}>r\]
% depending on which scenario it is in.

Conversely, advertiser 1 consistently enters the auction as it does not have an outside option. This necessitates that the platform opts for a pricing model ensuring advertiser 2's participation in the auction. If only advertiser 1 is present, the platform will invariably receive zero payoffs due to the second-price auction.
% On the other hand, advertiser 1 always enters the auction because it has no outside option. This implies that the platform must choose a pricing model that ensures that advertiser 2 enters the auction. If there is only advertiser 1, the platform will always have zero payoff due to the second-price auction.

% The outside option, $r$, may represent the payoff from another competitor's platform or the payoff from other traditional ad media. Advertisers with budget constraints will spend their budget on our focal platform if and only if its payoff is greater than other choices.

% On the other hand, advertiser 1 always enters the auction, representing those advertisers without any budget constraint. This implies the platform must choose a pricing model that ensures advertiser 2 enters the auction. If there is only advertiser 1, the platform will always have zero payoff because of the second-price auction. This setting reduces the model's complexity by utilizing two representative advertisers with real-world meanings, which highlights the main incentive of pricing models.

As a result, in the first stage, the platform selects a pricing model that guarantees advertiser 2's involvement in the auction first and then maximizes its payoffs. The optimal pricing model is
\begin{align*}
    X^*=argmax_X \pi^P_X \\
    s.t.\ \pi^2_X>r
\end{align*}

This equation indicates that if a platform impulsively shifts to a pricing model featuring higher platform payoffs but lower advertiser payoffs, it might mistakenly perceive an increase in payoffs through A/B testing. However, its payoffs will ultimately decline as some advertisers with budget constraints opt to leave the platform.

% This formula implies that if a platform rashly switches to a pricing model with higher platform payoffs but lower advertiser payoffs, it may erroneously observe an increase in payoffs through a/b testing, but its payoffs will decrease in the long run as some advertisers will eventually leave the platform.

Our analysis suggests that platforms maintain a balance between platform payoffs with a predetermined number of advertisers and advertiser payoffs that draw in more advertisers. When user value is substantial, advertisers will naturally join the platform, enabling pricing models with high platform payoffs. Conversely, platforms should adopt pricing models with high advertiser payoffs to enhance their payoffs in an alternative manner.

% Our analysis suggests that platforms should balance platform payoffs with a fixed number of advertisers and advertiser payoffs that attract more advertisers. When user value is high, advertisers will naturally enter the platform, allowing CPC or even CPM to further increase platform payoffs. In contrast, platforms should choose OCPC to attract more advertisers and increase their payoffs in a different way.

\subsection{Equilibrium}
% We further study how the outside option, $r$, affects the game. 

We further examine the game's equilibriums under various outside options $r$. Figure~\ref{fig:outside} illustrates each party's payoff under different outside options in the out-site scenario. Specifically, Figure~\ref{fig:outside} (1) to (4) show the platform's payoff, social welfare, advertiser 1's payoff, and advertiser 2's payoff, respectively.

% We further study our game's outcome under various outside options, $r$.

% Figure~\ref{fig:outside} shows each party's payoff under various outside options in the out-site scenario. Figure~\ref{fig:outside} (1) to (4) show the platform's payoff, the social welfare, advertiser 1's payoff, and advertiser 2's payoff respectively.

As seen in Figure~\ref{fig:outside}, when $r<\pi^2_{CPC}$, the platform selects CPC. When $\pi^2_{CPC}<r<\pi^2_{OCPC}$, the platform selects OCPC. For $r>\pi^2_{OCPC}$, the platform exhibits indifference between pricing models as the payoffs remain consistently zero.

% Figure~\ref{fig:outside} shows that when $r<\pi^2_{CPC}$, the platform will choose CPC. When $\pi^2_{CPC}<r<\pi^2_{OCPC}$, the platform will choose OCPC. When $r>\pi^2_{OCPC}$, the platform is indifferent in pricing models, because it always has zero profit.

The shaded region in vertical lines highlights the benefits from OCPC's innovation. Prior to OCPC's development, the platform would incur no payoff when $\pi^2_{CPC}<r<\pi^2_{OCPC}$, as advertiser 2 would abstain from entering the auction. This shaded region depicted in Figure~\ref{fig:outside} (1) signifies the enhancement in platform payoffs due to OCPC, indicating that developing a new pricing model that yields higher advertiser payoffs is advantageous for platforms.

% The shaded area in vertical lines highlights the advantages of OCPC's invention. Before the emergence of OCPC, the platform would receive zero payoff when $\pi^2_{CPC}<r<\pi^2_{OCPC}$, as advertiser 2 would not participate in the auction. The shaded area in Figure~\ref{fig:outside} (1) represents the increase in the platform's payoff attributable to OCPC, suggesting that developing a new pricing model that yields higher advertiser payoffs benefits advertising platforms.

% The shadow area in vertical lines shows the benefits of OCPC. For example, if there is no OCPC, when $\pi^2_{CPC}<r<\pi^2_{OCPC}$, the platform has zero payoff because advertiser 2 does not join the auction. The shadow area in Figure~\ref{fig:outside} (1) shows the increase in the platform's payoff due to OCPC, indicating that OCPC attracts more advertisers by offering a new pricing model with a higher advertiser payoff, which benefits advertising platforms in return.

Interestingly, not all parties benefit from OCPC. The dotted region in Figure~\ref{fig:outside} (3) signifies a decrease in advertiser 1's payoff. This outcome occurs because OCPC introduces more competitors for those who consistently participate in auctions. As competition intensifies, bidding becomes more aggressive, leading to higher costs and subsequently reducing payoffs for these advertisers.

% Interestingly, not everyone benefits from OCPC. The dotted area in Figure~\ref{fig:outside} (3) reflects a payoff decrease for advertiser 1. The reason is that OCPC brings more advertisers and thus more competitors for advertisers with no budget constraint who always participate in auctions. The more competitive bidding leads to a higher cost, which reduces their payoffs.

% (1) shows that the platform's payoff is decreasing in the outside option because advertiser 2 is more likely to leave the auction with a higher outside option. Though this fact seems a sad story for the platform, it actually hints at several business strategies to increase advertising platforms' payoff:

\subsection{Managerial Implications}

Firstly, platforms can increase their payoffs by appropriately shifting their focus from CTR and CVR prediction to improving CTR and CTR. Overemphasizing CTR and CVR prediction and ad ranking may lead to diminishing returns because there is an upper bound on the platform's payoff, in which each ad is ranked accurately based on its value. On the other hand, improving CTR and CVR boost payoffs in two ways: enhancing $C^i$ or $P^i$ directly increases both $\pi^P_{CPC}$ and $\pi^P_{OCPC}$ and attracts more advertisers by shifting $\pi^2_{CPC}$ and $\pi^2_{OCPC}$ to the right. To achieve this, platforms can expand their high net worth user base, make ads more contextually relevant to their services, and improve ad quality. The role of algorithms in facilitating these improvements remains a promising direction for future research.

% First, platforms can increase their payoffs by appropriately shifting their focus from CTR and CVR prediction to improving CTR and CTR. Spending too much effort on CTR and CVR prediction and ad ranking has diminishing returns because there is an upper bound on the platform's payoffs if each ad is accurately ranked according to its value. In contrast, improving CTR and CVR increases payoffs in two ways, because improving $C^i$ or $P^i$ not only directly increases $\pi^P_{CPC}$ and $\pi^P_{OCPC}$, but also attracts more advertisers by shifting $\pi^2_{CPC}$ and $\pi^2_{OCPC}$ to the right. To improve CTR and CVR, platforms can grow high net worth users, make ads more relevant to their context, i.e. their services, and improve the quality of the ads. How algorithms help in these directions remains a promising research direction.

Secondly, platforms can increase their payoffs by indirectly reducing advertisers' outside options. One potential approach involves offering subsidies or coupons to advertisers. For instance, if $r=\pi^2_{OCPC}$, advertiser 2 will not enter the auction, resulting in no payoff for the platform. However, providing a small subsidy, $\epsilon$, to advertiser 2 will result in their participation in the auction, and the platform will receive a payoff of $\pi^P_{OCPC}-\epsilon$. The primary challenge in delivering subsidies lies in identifying deserving advertisers and determining the minimum subsidy amount, given that advertisers are highly strategic in such games.

% Second, platforms can increase their payoffs by indirectly reducing advertisers' outside options. One possible method is to offer subsidies or coupons to advertisers. For example, if $r=\pi^2_{OCPC}$, advertiser 2 will not enter the auction, resulting in a zero payoff for the platform. However, if the platform offers a tiny subsidy, $\epsilon$, to advertiser 2, it will enter the auction and the platform will receive a payoff of $\pi^P_{OCPC}-\epsilon$. The main challenge in providing subsidies is to distinguish between advertisers who deserve the subsidy and to decide on the minimum subsidy to spend, considering that advertisers are absolutely strategic in such games.

\begin{figure}[!h]
\centering
\includegraphics[width=0.25\textwidth]{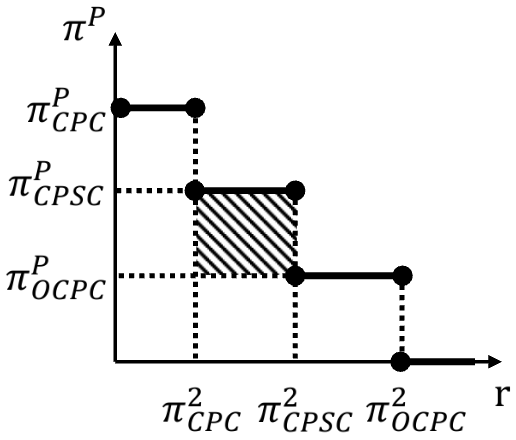}
% \vspace{-4mm}
\caption{Cost Per Shopping Cart (CPSC) Example}\label{fig:CPSC}
% \vspace{-5mm}
\end{figure}

Lastly, platforms can increase their payoffs by developing new pricing models wherein advertisers bid for intermediate outcomes between clicks and conversions. For example, we can develop a pricing model for online retailers called Cost Per Shopping Cart (CPSC), where advertisers bid for instances when a user adds a product to their shopping cart and they pay for clicks or impressions. As illustrated in Figure \ref{fig:CPSC}, CPSC attracts more advertisers than CPC and generates higher platform payoffs than OCPC. The shaded region represents its benefits.

% Third, platforms can increase their payoffs by developing new pricing models in which advertisers bid for something between clicks and conversions. For example, we can develop a pricing model for online retailers called Cost Per Shopping Cart (CPSC), in which advertisers bid for each time a user adds a product to its shopping cart and pay for clicks or impressions. As shown in Figure 2, CPSC attracts more advertisers than CPC and generates higher platform payoffs than OCPC. The slashed area represents its benefit.
\section{Conclusion}
In conclusion, this paper pioneers the study of the Optimized Cost Per Click (OCPC) pricing model in online advertising from an economic perspective. This paper identifies two key factors contributing to OCPC's prevalence: its versatility in out-site scenarios when advertisers can manipulate conversion reporting and its ability to attract more advertisers with outside options. This paper sheds light on the optimal pricing model selection for online advertising platforms and offers suggestions to further improve their payoffs.

% This paper pioneers the study of the Optimized Cost Per Click (OCPC) pricing model in online advertising from an economic perspective. This paper demonstrates the conditions under which OCPC can outperform other traditional pricing models. Finally, this paper sheds light on the optimal pricing model selection for online advertising platforms and offers suggestions to further improve their payoffs.

\begin{acks}
This work was inspired during Kaichen Zhang's internship at Bytedance. Thanks to Hongyi Zhang, Dingkun Hong, and the Bytedance Adscore team for their help.
\end{acks}

\newpage

\bibliographystyle{ACM-Reference-Format}
\bibliography{reference}

% \clearpage
\raggedbottom
\newpage

\appendix
\section{appendix}
\subsection{Other Pricing Models} \label{app:other}
The definition of CPM and OCPM is included in this section. Their analysis is analogous to CPC and OCPC.
\subsubsection{\textbf{Cost Per Mille}}
In CPM, advertisers bid for impressions and pay for impressions. 

As advertisers bid for impressions, the platform predicts advertiser $i$'s equivalent bid on an impression as $e^i=b^i$.

As advertisers pay for impressions, the auction winner has to pay $b^l$ after the ad display. The loser pays 0.

\subsubsection{\textbf{Optimized Cost Per Mille}}
In OCPM, advertisers bid for conversions but pay for impressions. 

As advertisers bid for conversions, the platform's prediction on advertiser $i$'s equivalent bid on an impression is $e^i=\bar{c^i}\bar{p^i}b^i$. 

As advertisers pay for impressions, the auction winner has to pay $\bar{c^l}\bar{p^l}b^l$ immediately after the ad display. The loser pays 0.

\kc{
\subsection{Assumption Justification} \label{app:assumption}
We justify the assumptions made in this paper:

\begin{enumerate} 
\item Advertiser 1 has no outside option, while advertiser 2 has an outside option $r$: refer to Section \ref{sec:setting}.
\item The platform’s decisions are based on its predictions of $c^i$ and $p^i$: In the real world, ads are ranked by platforms according to CTR/CVR predictions.
\item Advertisers make decisions based on the distribution of $C^i$ and $P^i$ ...: In the real world, advertisers make decisions based on their CTR/CVR, evidenced by various CTR/CVR statics in advertiser panels offered by platforms.
\item ... and their belief in the accuracy of the platform's predictions: Advertisers can't know the real value of the platform's predictions and can only form belief instead.
\item The platform's predictions are accurate, with $\bar{c^i}=c^i$ and $\bar{p^i}=p^i$: Platforms in the real world optimize their predictions to be accurate. This assumption is made for the sake of mathematical convenience.
\item ... and advertisers' beliefs are consistent, with $\bar{C^i}=C^i$ and $\bar{P^i}=P^i$: advertisers should form a Rational Expectation, consistent with the truth. 
\item $\bar{c^i}=c^i$, $\bar{p^i}=\alpha^ip^i$, $\bar{C^i}=C^i$ and $\bar{P^i}=\alpha^iP^i$: Previous assumptions' out-site scenario version.
\end{enumerate}
}

\subsection{Proof of Proposition \ref{pro:CPA}} \label{app:CPA}
We show that $\pi^i(m^i|e^k)\ge\pi^i(b^i|e^k)$ for any $e^k$.

When $b^i>m^i$,
\begin{align*}
&\pi^i(b^i|e^k)
\\=&
Pr(\bar{C^i}\bar{P^i}m^i>e^k)E(C^iP^im^i-C^iP^i \frac{e^k}{\bar{C^i}\bar{P^i}}|\bar{C^i}\bar{P^i}m^i>e^k)
\\
&+Pr(\bar{C^i}\bar{P^i}b^i>e^k>\bar{C^i}\bar{P^i}m^i)\\&E(C^iP^im^i-C^iP^i \frac{e^k}{\bar{C^i}\bar{P^i}}|\bar{C^i}\bar{P^i}b^i>e^k>\bar{C^i}\bar{P^i}m^i)
\\
<=&
Pr(\bar{C^i}\bar{P^i}m^i>e^k)E(C^iP^im^i-C^iP^i \frac{e^k}{\bar{C^i}\bar{P^i}}|\bar{C^i}\bar{P^i}m^i>e^k) = \pi^i(m^i|e^k)
\end{align*}
because $E(C^iP^im^i-C^iP^i \frac{e^k}{\bar{C^i}\bar{P^i}}|\bar{C^i}\bar{P^i}b^i>e^k>\bar{C^i}\bar{P^i}m^i)<0$.

When $b^i<m^i$,
\begin{align*}
&\pi^i(m^i|e^k)
\\=&
Pr(\bar{C^i}\bar{P^i}b^i>e^k)E(C^iP^im^i-C^iP^i \frac{e^k}{\bar{C^i}\bar{P^i}}|\bar{C^i}\bar{P^i}b^i>e^k)
\\
&+Pr(\bar{C^i}\bar{P^i}b^i<e^k<\bar{C^i}\bar{P^i}m^i)\\&E(C^iP^im^i-C^iP^i \frac{e^k}{\bar{C^i}\bar{P^i}}|\bar{C^i}\bar{P^i}b^i<e^k<\bar{C^i}\bar{P^i}m^i)
\\
>=&
Pr(\bar{C^i}\bar{P^i}b^i>e^k)E(C^iP^im^i-C^iP^i \frac{e^k}{\bar{C^i}\bar{P^i}}|\bar{C^i}\bar{P^i}b^i>e^k) \\=& \pi^i(b^i|e^k)
\end{align*}
So $\pi^i(m^i|e^k)>=\pi^i(b^i|e^k)$ for any $e^k$.

\subsection{Proof of Proposition \ref{pro:OCPC}} \label{app:OCPC}
We show that $\pi^i(m^i|e^k)\ge\pi^i(b^i|e^k)$ for any $e^k$.

When $b^i>m^i$,
\begin{align*}
&\pi^i(b^i|e^k)
\\=&
Pr(\bar{C^i}\bar{P^i}m^i>e^k)E(C^iP^im^i-C^i \frac{e^k}{\bar{C^i}}|\bar{C^i}\bar{P^i}m^i>e^k)
\\&
+Pr(\bar{C^i}\bar{P^i}b^i>e^k>\bar{C^i}\bar{P^i}m^i)\\&E(C^iP^im^i-C^i \frac{e^k}{\bar{C^i}}|\bar{C^i}\bar{P^i}b^i>e^k>\bar{C^i}\bar{P^i}m^i)
\\
<=&
Pr(\bar{C^i}\bar{P^i}m^i>e^k)E(C^iP^im^i-C^i \frac{e^k}{\bar{C^i}}|\bar{C^i}\bar{P^i}m^i>e^k) \\=& \pi^i(m^i|e^k)
\end{align*}
\vspace{-3mm}
When $b^i<m^i$,
\vspace{+1mm}
\begin{align*}
\pi^i(m^i|e^k)&
% \\=&
=
Pr(\bar{C^i}\bar{P^i}b^i>e^k)E(C^iP^im^i-C^i \frac{e^k}{\bar{C^i}}|\bar{C^i}\bar{P^i}b^i>e^k)
\\&
+Pr(\bar{C^i}\bar{P^i}b^i<e^k<\bar{C^i}\bar{P^i}m^i)\\&E(C^iP^im^i-C^i \frac{e^k}{\bar{C^i}}|\bar{C^i}\bar{P^i}b^i<e^k<\bar{C^i}\bar{P^i}m^i)
% \\
\end{align*}
\begin{align*}
>=&
Pr(\bar{C^i}\bar{P^i}b^i>e^k)E(C^iP^im^i-C^i \frac{e^k}{\bar{C^i}}|\bar{C^i}\bar{P^i}b^i>e^k)  \\=& \pi^i(b^i|e^k)
\end{align*}
So $\pi^i(m^i|e^k)>=\pi^i(b^i|e^k)$ for any $e^k$.

\subsection{Proof of Proposition \ref{pro:i_CPA}} \label{app:i_CPA}
Suppose there exists an equilibrium, where advertiser $i$'s reporting strategy is $\alpha^{i*}$. Then, there is $\bar{p^i}=\alpha^{i*}p^i$ and $\bar{P^i}=\alpha^{i*}P^i$ accordingly in this equilibrium.

If $\alpha^{i*}=0$, then $\bar{p^i}=0$ and $e^i =0$. Advertiser $i$ will always lose the auction and have a $0$ utility. Thus, advertiser $i$ can deviate to a positive $\alpha^{i*}$ to have a positive utility.

If $\alpha^{i*}>0$,
we show that given an $\alpha^i>0$, the optimal $b_i=\frac{m^i}{\alpha^i}$, i.e., $\pi^i(\frac{m^i}{\alpha^i},\alpha^i|e^k)\ge\pi^i(b^i,\alpha^i|e^k)$ for any $e^k$.

When $b^i>\frac{m^i}{\alpha^i}$,
\begin{align*}
&\pi^i(b^i,\alpha^i|e^k)
\\&=
Pr(\bar{C^i}\bar{P^i}\frac{m^i}{\alpha^i}>e^k)E(C^iP^im^i-C^i\alpha^iP^i \frac{e^k}{\bar{C^i}\bar{P^i}}|\bar{C^i}\bar{P^i}\frac{m^i}{\alpha^i}>e^k)
\\&
+Pr(\bar{C^i}\bar{P^i}b^i>e^k>\bar{C^i}\bar{P^i}\frac{m^i}{\alpha^i})
\\&E(C^iP^im^i-C^i\alpha^iP^i \frac{e^k}{\bar{C^i}\bar{P^i}}|\bar{C^i}\bar{P^i}b^i>e^k>\bar{C^i}\bar{P^i}\frac{m^i}{\alpha^i})
\\&
<=
Pr(\bar{C^i}\bar{P^i}\frac{m^i}{\alpha^i}>e^k)E(C^iP^im^i-C^i\alpha^iP^i \frac{e^k}{\bar{C^i}\bar{P^i}}|\bar{C^i}\bar{P^i}\frac{m^i}{\alpha^i}>e^k) 
\\&
= \pi^i(\frac{m^i}{\alpha^i},\alpha^i|e^k)
\end{align*}

When $b^i<\frac{m^i}{\alpha^i}$,
\begin{align*}
&\pi^i(\frac{m^i}{\alpha^i},\alpha^i|e^k)
\\&=
Pr(\bar{C^i}\bar{P^i}b^i>e^k)E(C^iP^im^i-C^i\alpha^iP^i \frac{e^k}{\bar{C^i}\bar{P^i}}|\bar{C^i}\bar{P^i}b^i>e^k)
\\&
+Pr(\bar{C^i}\bar{P^i}b^i<e^k<\bar{C^i}\bar{P^i}\frac{m^i}{\alpha^i})
\\&
E(C^iP^im^i-C^i\alpha^iP^i \frac{e^k}{\bar{C^i}\bar{P^i}}|\bar{C^i}\bar{P^i}b^i<e^k<\bar{C^i}\bar{P^i}\frac{m^i}{\alpha^i})
\\&
>=
Pr(\bar{C^i}\bar{P^i}\frac{m^i}{\alpha^i}>e^k)
E(C^iP^im^i-C^i\alpha^iP^i \frac{e^k}{\bar{C^i}\bar{P^i}}|\bar{C^i}\bar{P^i}\frac{m^i}{\alpha^i}>e^k) 
\\&
= \pi^i(b^i,\alpha^i|e^k)    
\end{align*}
So $\pi^i(\frac{m^i}{\alpha^i},\alpha^i|e^k)\ge\pi^i(b^i,\alpha^i|e^k)$ for any $e^k$.

Then, when $b_i=\frac{m^i}{\alpha^i}$,
\begin{align*}
&\pi^i(\alpha^i|e^k)
\\&= 
Pr(\bar{C^i}\bar{P^i}\frac{m^i}{\alpha^i}>e^k)E(C^iP^im^i-C^i\alpha^iP^i \frac{e^k}{\bar{C^i}\bar{P^i}}|\bar{C^i}\bar{P^i}\frac{m^i}{\alpha^i}>e^k)
\\&=
Pr(\bar{C^i}\alpha^{i*}P^i\frac{m^i}{\alpha^i}>e^k)E(C^iP^im^i-\frac{\alpha^i}{\alpha^{i*}}e^k|\bar{C^i}\alpha^{i*}P^i\frac{m^i}{\alpha^i}>e^k)
\end{align*}
Note that $\pi_i$ increases when $\alpha^i$ decreases.
By deviating to $\alpha^i=\epsilon\alpha^{i*}$ where $0<\epsilon<1$, advertiser $i$ gets a higher utility.

\subsection{Proof of Proposition \ref{pro:i_OCPC}} \label{app:i_OCPC}
We first show that $\alpha^i$ and $b^i$ always appear together. So we represent $\pi^i(b^i,\alpha^i|e^k)$ by $\pi^i(\alpha^ib^i|e^k)$ for convenience.
\begin{align*}
&\pi^i(b^i,\alpha^i|e^k)= Pr(\bar{C^i}\bar{P^i}b^i>e^k)E(C^iP^im^i-C^i \frac{e^k}{\bar{C^i}}|\bar{C^i}\bar{P^i}b^i>e^k)
\\&=
Pr(\bar{C^i}P^i\alpha^ib^i>e^k)E(C^iP^im^i-C^i \frac{e^k}{\bar{C^i}}|\bar{C^i}P^i\alpha^ib^i>e^k)
\\&=\pi^i(\alpha^ib^i|e^k)
\end{align*}
We then show that $\pi^i(m^i|e^k)\ge\pi^i(\alpha^ib^i|e^k)$ for any $e^k$.

When $\alpha^ib^i>m^i$,
\begin{align*}
&\pi^i(\alpha^ib^i|e^k)
\\&=
Pr(\bar{C^i}P^im^i>e^k)E(C^iP^im^i-C^i \frac{e^k}{\bar{C^i}}|\bar{C^i}P^im^i>e^k)
\\&
+Pr(\bar{C^i}P^i\alpha^ib^i>e^k>\bar{C^i}P^im^i)
\\&E(C^iP^im^i-C^i \frac{e^k}{\bar{C^i}}|\bar{C^i}P^i\alpha^ib^i>e^k>\bar{C^i}P^im^i)
\\&
<=
Pr(\bar{C^i}P^im^i>e^k)E(C^iP^im^i-C^i \frac{e^k}{\bar{C^i}}|\bar{C^i}P^im^i>e^k) = \pi^i(m^i|e^k)
\end{align*}
When $\alpha^ib^i<m^i$,
\begin{align*}
&\pi^i(m^i|e^k)
\\&=
Pr(\bar{C^i}P^i\alpha^ib^i>e^k)E(C^iP^im^i-C^i \frac{e^k}{\bar{C^i}}|\bar{C^i}P^i\alpha^ib^i>e^k)
\\&
+Pr(\bar{C^i}P^i\alpha^ib^i<e^k<\bar{C^i}P^im^i)
\\&E(C^iP^im^i-C^i \frac{e^k}{\bar{C^i}}|\bar{C^i}P^i\alpha^ib^i<e^k<\bar{C^i}P^im^i)
\\&
>=
Pr(\bar{C^i}P^i\alpha^ib^i>e^k)E(C^iP^im^i-C^i \frac{e^k}{\bar{C^i}}|\bar{C^i}P^i\alpha^ib^i>e^k) 
\\&= \pi^i(\alpha^ib^i|e^k)
\end{align*}

\kc{
\subsection{Proof of Lemma \ref{lemma:social}} \label{app:lemma_social}
\begin{align*}
&E^S_{OCPC}=\sum_{i=1}^2 E(C^iP^im^i|i=argmax_k C^kP^km^k)Pr(i=argmax_k C^kP^km^k)
\\&E^S_{CPC}=\sum_{j=1}^2 E(C^jP^jm^j|j=argmax_k C^k\mu^p_km^k)Pr(j=argmax_k C^k\mu^p_km^k)
\\&E^S_{OCPC}-E^S_{CPC}=
\\&\sum_{i,j\ne i}E(C^iP^im^i-C^jP^jm^j|i=argmax_k C^kP^km^k,j=argmax_k C^k\mu^p_km^k)
\\&>0
\end{align*}
because $i=argmax_k C^kP^km^k$.

\subsection{Proof of Lemma \ref{lemma:plat}} \label{app:lemma_plat}
\begin{align*}
&E^P_{OCPC}=\sum_{i=1}^2 E(C^iP^im^i|i=argmin_k C^kP^km^k)Pr(i=argmin_k C^kP^km^k)
\\&E^P_{CPC}=\sum_{j=1}^2 E(C^jP^jm^j|j=argmin_k C^k\mu^p_km^k)Pr(j=argmin_k C^k\mu^p_km^k)
\\&E^P_{OCPC}-E^P_{CPC}=
\\&\sum_{i,j\ne i}E(C^iP^im^i-C^jP^jm^j|i=argmin_k C^kP^km^k,j=argmin_k C^k\mu^p_km^k)
\\&<0
\end{align*}
because $i=argmin_k C^kP^km^k$.

\subsection{Proof of Lemma \ref{lemma:adv}} \label{app:lemma_adv}
If advertiser $i$ wins in both OCPC and CPC, its payoffs are the same ($C^iP^im^i-C^kP^km^k$) in both OCPC and CPC. If advertiser $i$ loses in both OCPC and CPC, its payoffs are also the same ($0$). The contribution of such cases to $\pi^i_{OCPC}-\pi^i_{CPC}$ is 0. As a result, 
\begin{align*}
&\pi^i_{OCPC}-\pi^i_{CPC}=
% \\&E(C^iP^im^i-C^kP^km^k|C^iP^im^i>C^kP^km^k)Pr(C^iP^im^i>C^kP^km^k)
% \\&E(C^i\mu^p_im^i-C^k\mu^p_km^k|C^i\mu^p_im^i>C^k\mu^p_km^k)Pr(C^i\mu^p_im^i>C^k\mu^p_km^k)
\\&E(-(C^iP^im^i-C^kP^km^k)|C^iP^im^i<C^kP^km^k,C^i\mu^p_im^i>C^k\mu^p_km^k)
\\&Pr(C^iP^im^i<C^kP^km^k,C^i\mu^p_im^i>C^k\mu^p_km^k) +
\\&E(C^iP^im^i-C^kP^km^k|C^iP^im^i>C^kP^km^k,C^i\mu^p_im^i<C^k\mu^p_km^k)
\\&Pr(C^iP^im^i>C^kP^km^k,C^i\mu^p_im^i<C^k\mu^p_km^k)>0
% \\&E(C^i\mu^p_im^i-C^k\mu^p_km^k|C^i\mu^p_im^i>C^k\mu^p_km^k)Pr(C^i\mu^p_im^i>C^k\mu^p_km^k)
\end{align*}
because both terms are positive.
}

\subsection{Extensions} \label{app:extensions}
\subsubsection{\textbf{Repeated Auctions}} We extend our model to the repeated auction setting, i.e., multiple auctions are in the game. Specifically, the timeline of the new game is as follows:

\begin{enumerate}
    \item The platform selects a model among CPC, CPA, or OCPC.
    \item Advertiser 2 determines whether to enter the auction.
    \item Advertisers select their bids and, if in the out-site scenario, also decide on their reporting strategy $\alpha^i$.
    \item The actual values of $c^i$ and $p^i$ are realized, while the platform predicts $\bar{c^i}$ and $\bar{p^i}$, and calculates $e^i$.
    \item The platform identifies the auction winner ($w=argmax_i e^i$) and display its ad.
    \item The winner makes payments and all parties realize payoffs.
    \item The steps (3) to (6) will be repeated $T$ times.
\end{enumerate}
Compared to the original game, the new game now has $T$ auctions. The overall utility of the game is the sum of the utility of each auction.

We show that the repeated auction setting has the same bidding strategies as the original setting.

Using CPC as an example, the utility of advertiser $i$ is:
\begin{equation*}
\sum^T_{t=1}\pi^i_t(b^i_t|e^k_t)
= \sum^T_{t=1} Pr(\bar{C^i}b^i_t>e^k_t)E(C^iP^im^i-C^i \frac{e^k_t}{\bar{C^i}}|\bar{C^i}b^i_t>e^k_t)
\end{equation*}
where another advertiser $k$'s equivalent bid on an impression in auction $t$ is $e^k_t$, and advertiser $i$ bid in auction $t$ is $b^i_t$.

We notice that bidding $b^i_t=\mu^p_im^i$ is the dominant strategy in each auction, regardless of how others bid ($e^k_t$). As a result, we have:
\begin{equation*}
max\sum^T_{t=1}\pi^i_t(b^i_t|e^k_t)=\sum^T_{t=1}\pi^i_t(\mu^p_im^i|e^k_t)
= T\pi^i(\mu^p_im^i|e^k)
\end{equation*}
where advertiser $i$'s utility in the repeated auctions setting is exactly $T$ times the original setting. Consequently, the platform's utility is also $T$ times as much as before.

The above analysis demonstrates that the analysis applied to the original setting also applies to the repeated auctions setting, since each agent's utility is scaled up by a factor of $T$.  

\newpage
\subsubsection{\textbf{Multiple Advertisers}} We extend our model to the multiple advertiser setting, where there are $N$ advertisers in the game. 

We show that our main insights also hold in the multiple advertiser setting: for online advertising platforms, (1) OCPC dominates CPA; (2) OCPC does not dominate CPC. 

We define the highest equivalent bid on an impression of other advertisers as $e^{-i}$ for advertiser $i$. By its definition, 
\begin{equation*}
e^{-i}=max_{k\ne i}e^k
\end{equation*}

In the in-site scenario, advertiser $i$'s utility in CPA and OCPC are respectively
\begin{align*}
&\pi^i_{I-CPA}(b^i|e^{-i})=\\ &Pr(\bar{C^i}\bar{P^i}b^i>e^{-i})E(C^iP^im^i-C^iP^i \frac{e^{-i}}{\bar{C^i}\bar{P^i}}|\bar{C^i}\bar{P^i}b^i>e^{-i})
\end{align*}
and
\begin{align*}
&\pi^i_{I-OCPC}(b^i|e^{-i})=\\ &Pr(\bar{C^i}\bar{P^i}b^i>e^{-i})E(C^iP^im^i-C^i \frac{e^{-i}}{\bar{C^i}}|\bar{C^i}\bar{P^i}b^i>e^{-i})
\end{align*}

Similarly, in the out-site scenario, advertiser $i$'s utility in CPA and OCPC are respectively 
\begin{align*}
&\pi^i_{O-CPA}(b^i,\alpha^i|e^{-i})=\\ &Pr(\bar{C^i}\bar{P^i}b^i>e^{-i})E(C^iP^im^i-C^i\alpha^iP^i \frac{e^{-i}}{\bar{C^i}\bar{P^i}}|\bar{C^i}\bar{P^i}b^i>e^{-i})
\end{align*}
and
\begin{align*}
&\pi^i_{O-OCPC}(b^i,\alpha^i|e^{-i})=\\ &Pr(\bar{C^i}\bar{P^i}b^i>e^{-i})E(C^iP^im^i-C^i \frac{e^{-i}}{\bar{C^i}}|\bar{C^i}\bar{P^i}b^i>e^{-i})
\end{align*}

Note that the utility of advertiser $i$ in the multiple advertiser setting can be derived by simply replacing $e^k$ with $e^{-i}$. This is because the second-price auction format ensures that whether advertiser $i$ wins or loses, and how much it pays, depends only on the highest one among the others. Consequently, the proofs in the original setting also apply to the multiple advertiser setting by replacing $e^k$ with $e^{-i}$.

As a result, OCPC also dominates CPA, given its broader applicability to out-site scenarios while maintaining the same performance in in-site scenarios as CPA.

To show that OCPC does not dominate CPC, we should find a counterexample where CPC leads to a higher payoff than OCPC for platforms. Obviously, by setting $N=2$, we have found such a counterexample in this paper's main body.

\clearpage
\raggedbottom

\end{document}